\documentclass[12pt]{article}
\usepackage{amsmath}
\usepackage{graphicx,psfrag,epsf}
\usepackage{enumerate}
\usepackage{natbib}
\usepackage{textcomp}
\usepackage[hyphens]{url} % not crucial - just used below for the URL
\usepackage{hyperref}

\newcommand{\blind}{0}

\providecommand{\tightlist}{%
  \setlength{\itemsep}{0pt}\setlength{\parskip}{0pt}}

\usepackage{longtable,booktabs,array}
\usepackage{calc} % for calculating minipage widths
\usepackage{etoolbox}
\makeatletter
\patchcmd\longtable{\par}{\if@noskipsec\mbox{}\fi\par}{}{}
\makeatother
\IfFileExists{footnotehyper.sty}{\usepackage{footnotehyper}}{\usepackage{footnote}}
\makesavenoteenv{longtable}

\usepackage{amsmath}
\usepackage{amsfonts}
\usepackage{amssymb}
\usepackage{algorithm}
\usepackage{algpseudocode}
\usepackage{mathtools}
\usepackage{bm}
\usepackage[dvipsnames]{xcolor}
\usepackage{hyperref}

\usepackage{amsthm}
\newtheorem{theorem}{Theorem}[section]
\newtheorem{lemma}{Lemma}[section]
\newtheorem{corollary}{Corollary}[section]

\theoremstyle{definition}
\newtheorem{definition}{Definition}[section]
\theoremstyle{definition}

\theoremstyle{definition}

\theoremstyle{definition}

\theoremstyle{remark}
\newtheorem*{remark}{Remark}

\begin{document}

\def\spacingset#1{\renewcommand{\baselinestretch}%
{#1}\small\normalsize} \spacingset{1}

%%%%%%%%%%%%%%%%%%%%%%%%%%%%%%%%%%%%%%%%%%%%%%%%%%%%%%%%%%%%%%%%%%%%%%%%%%%%%%

\if0\blind
{
  \title{\bf Generative Filtering for Recursive Bayesian Inference with Streaming Data}

  \author{
        Ian Taylor \thanks{Work performed while at Colorado State University} \\
    National Renewable Energy Laboratory\\
     and \\     Andee Kaplan \\
    Department of Statistics, Colorado State University\\
     and \\     Brenda Betancourt \\
    NORC at the University of Chicago\\
      }
  \maketitle
} \fi

\if1\blind
{
  \bigskip
  \bigskip
  \bigskip
  \begin{center}
    {\LARGE\bf Generative Filtering for Recursive Bayesian Inference with Streaming Data}
  \end{center}
  \medskip
} \fi

\bigskip
\begin{abstract}
In the streaming data setting, where data arrive continuously or in frequent batches and there is no pre-determined amount of total data, Bayesian models can employ recursive updates, incorporating each new batch of data into the model parameters' posterior distribution. Filtering methods are currently used to perform these updates efficiently, however, they suffer from eventual degradation as the number of unique values within the filtered samples decreases. We propose Generative Filtering, a method for efficiently performing recursive Bayesian updates in the streaming setting. Generative Filtering retains the speed of a filtering method while using parallel updates to avoid degenerate distributions after repeated applications. We derive rates of convergence for Generative Filtering and conditions for the use of sufficient statistics instead of fully storing all past data. We investigate the alleviation of filtering degradation through simulation and an ecological time series of counts.
\end{abstract}

\noindent%

\vfill

\newpage
\spacingset{1.75} % DON'T change the spacing!

\hypertarget{introduction}{%
\section{Introduction}\label{introduction}}

Modern data collection has given rise to the streaming setting, where data arrive continuously or in frequent batches. In a typical analysis, estimates of model parameters can be produced in one offline procedure. However, in the streaming data setting, estimates of model parameters may be desired after each arrival of new data. This setting poses a computational challenge as the data model becomes more complex with each new arrival and producing model estimates is increasingly more time-consuming.

In Bayesian statistics, parameters are estimated using the posterior distribution of the parameters conditioned on the values of the data. Because these posterior distributions are frequently intractable, they are typically approximated using samples produced by Markov chain Monte Carlo \citep[MCMC,][]{gelfand1990sampling}. MCMC can be computationally intensive, especially if a model is complex or has strong dependencies among the parameters. Sampling techniques such as Hamiltonian Monte Carlo \citep{duane1987hybrid} or the No-U-Turn Sampler \citep{hoffman2014no} more efficiently sample from the posterior distribution. However even with efficient sampling, offline MCMC begins from scratch each time new data arrive. This frequent restarting in the streaming setting ignores the previous parameter estimates, which may be helpful to more efficiently produce inference.

Approximate methods such as Variational Bayes \citep{blei2017variational} or Approximate Bayesian Computation \citep[ABC,][]{tavare1997inferring} provide Bayesian inference by approximating the posterior distribution. Variational Bayes approximates the distribution analytically, and ABC facilitates sampling from an approximation to the posterior distribution. These methods improve efficiency for intractable posterior distributions, but while the posterior means may be captured by the approximation, the entire shape of the posterior distribution is not. A practitioner may not be willing to make this trade-off between efficiency and inference. Additionally, these methods require the selection of a family of approximate posteriors (in the case of Variational Bayes) or summary statistics of data (in the case of ABC) for accurate approximation of the posterior, which may not be obvious.

Recursive Bayesian updates, in which the posterior distribution of one analysis is used as the prior for a subsequent analysis, have been proposed to leverage the streaming setting \citep{sarkka2013bayesian}. In the streaming setting, the posterior after the arrival of data up to time \(t\) is used as the prior when the next data arrive at time \(t+1\). When models are designed with conjugate prior distributions, recursive Bayesian updates can be performed analytically. Otherwise, samples approximating the previous posterior distribution can be used to inform the next analysis.

Sequential Monte Carlo (SMC) methods apply importance sampling to reweight and filter samples to changing probability distributions, particularly posterior distributions of state space models, and include sequential importance sampling \citep[e.g.,][]{hendry1992likelihood, liu1995blind}, the Bootstrap filter \citep{gordon1993novel}, and Monte Carlo filtering \citep{kitagawa1996monte}. SMC is applicable to the streaming setting through extension to sequences of arbitrary probability distributions \citep{neal2001annealed, chopin2002sequential, delmoral2006sequential}. \citet{berzuini1997dynamic} combine MCMC kernels and filtering to support dynamic models with expanding parameter spaces. Because these methods refine existing samples, we refer to them as ``filtering'' methods. Filtering methods are desirable for their speed but suffer from eventual degeneracy as samples are filtered. Markov chains have been combined with SMC methods to avoid filtering degradation \citep[e.g.][]{maceachern1999sequential, gilks2001following, dau2022waste}. However, the specific SMC samplers used may not be appropriate for all applications, e.g.~high dimensional discrete parameters spaces pose challenges for the choice of proposal.

We focus on two MCMC strategies for recursive Bayesian updates: Sequential MCMC \citep[SMCMC,][]{yang2013sequential} and Prior-Proposal-Recursive Bayes \citep[PPRB,][]{hooten2021making}. SMCMC uses an existing sample as initial values for independent Markov chains targeting the updated posterior. PPRB is a multistage MCMC method that filters existing samples through independent Metropolis-Hastings (MH) steps based on new data.

In this paper, we propose Generative Filtering, a method for recursive Bayesian inference in a streaming data setting that retains the speed benefits of filtering while avoiding the associated sample degradation. In Section \ref{sec:pprb-degradation} we derive theoretical bounds for the approximation error introduced through successive applications of PPRB. In Section \ref{sec:generative-filtering} we define Generative Filtering and describe Generative Filtering's relationship to the existing methods of SMCMC and PPRB. We also provide theoretical upper bounds on the convergence of Generative Filtering to the target posterior, and derive conditions to reduce data storage requirements. In Section \ref{sec:simulation-studies} we apply Generative Filtering to two numerical examples and evidence the inferential and computational benefits. In Section \ref{sec:application} we apply Generative Filtering to a publicly available data set of Steller sea lion (\emph{Eumetopias jubatus}) pup counts in Alaska, USA. We next provide necessary background information and notation that ground our development of Generative Filtering.

\hypertarget{sequential-markov-chain-monte-carlo}{%
\subsection{Sequential Markov Chain Monte Carlo}\label{sequential-markov-chain-monte-carlo}}

Sequential MCMC (SMCMC) is an algorithm to sample from a sequence of probability distributions, corresponding to posterior distributions at different times in streaming data contexts \citep{yang2013sequential}. Consider a model,
\begin{equation}
\bm{y}_1 \sim p(\bm{y}_1 | \bm{\theta}), \quad 
\bm{y}_2 \sim p(\bm{y}_2 | \bm{\theta}, \bm{\phi}, \bm{y}_1), \quad
\bm{\theta} \sim p(\bm{\theta}), \quad
\bm{\phi} \sim p(\bm{\phi} | \bm{\theta}),
\label{eqn:streaming-model-basic}
\end{equation}
where \(\bm{y}_1\) is available first and \(\bm{y}_2\) is available later. The parameters, \(\bm{\theta}\), encode the distribution of all data, while the parameters, \(\bm{\phi}\), encode only the distribution of \(\bm{y}_2\). SMCMC operates on an ensemble of samples, \(\{\bm{\theta}_s\}_{s=1}^S\), from the posterior distribution, \(p(\bm{\theta}|\bm{y}_1)\). First, a jumping kernel is applied to each sample \(\bm{\theta}_s\) in parallel to append a value \(\bm{\phi}_s\). Then, in parallel with each pair \((\bm{\theta}_s, \bm{\phi}_s)\) as an initial value, a transition kernel is applied \(m_t\) times. The transition kernel targets the posterior distribution \(p(\bm{\theta}, \bm{\phi}|\bm{y}_1, \bm{y}_2)\). Namely, for each initialization value \((\bm{\theta}_s, \bm{\phi}_s)^{(0)} = (\bm{\theta}_s, \bm{\phi}_s)\) for \(i = 1, \dots, m_t\) a sample is drawn \((\bm{\theta}_s, \bm{\phi}_s)^{(i)}\) from the density \(p\left(\bm{\theta}, \bm{\phi}|\bm{y}_1, \bm{y}_2, \bm{\theta}_s^{(i-1)}, \bm{\phi}_s^{(i-1)}\right) = T\left(\{\bm{\theta}_s^{(i-1)}, \bm{\phi}_s^{(i-1)}\}, \{\bm{\theta}, \bm{\phi}\}\right)\). The transition kernel \(T\) is defined so that the posterior \(p(\bm \theta, \phi|\bm y_1 \bm y_2)\) is the stationary measure of the Markov chain with transition kernel \(T\), i.e.~
\[
p(\bm \theta^{'}, \phi^{'}|\bm y_1 \bm y_2) = \int p(\bm \theta, \phi|\bm y_1 \bm y_2) T(\{\bm \theta, \phi\}, \{\bm \theta^{'}, \phi^{'}\}) d\{\bm \theta, \phi\}.
\]
After \(m_t\) iterations, the final value of each of the \(S\) parallel chains, \(\left\{(\bm{\theta}_s, \bm{\phi}_s)^{(m_t)}\right\}_{s = 1}^S\), is saved and comprise a sample to approximate the posterior distribution \(p(\bm{\theta}, \bm{\phi}|\bm{y}_1,\bm{y}_2)\).

\hypertarget{prior-proposal-recursive-bayes}{%
\subsection{Prior-Proposal-Recursive Bayes}\label{prior-proposal-recursive-bayes}}

Prior-Proposal-Recursive Bayes (PPRB) is a method for performing recursive Bayesian updates using existing samples from the previous posterior distribution \citep{hooten2021making}. Consider the general model in Eq. (\ref{eqn:streaming-model-basic}), but simplified so that there is no parameter \(\bm{\phi}\) and both data distributions are characterized by \(\bm{\theta}\). In PPRB the posterior samples, \(\{\bm{\theta}_s\}_{s=1}^S\), from \(p(\bm{\theta}|\bm{y}_1)\) are used as independent MH proposals in an MCMC to sample from the target \(p(\bm{\theta}|\bm{y}_1, \bm{y}_2)\), where the MH acceptance ratio for a proposal \(\bm{\theta}^\ast\) and a current value \(\bm{\theta}\) simplifies to
\begin{equation*}
\alpha = \min\left(\frac{p(\bm{y_2}|\bm{\theta}^\ast, \bm{y}_1)}{p(\bm{y_2}|\bm{\theta}, \bm{y}_1)}, 1\right).
\end{equation*}

One limitation of this method is its inability to account for a changing parameter space in a Bayesian update, represented in Eq. (\ref{eqn:streaming-model-basic}) by the parameter \(\bm{\phi}\). \citet{hooten2021making} propose first drawing \(\bm{\phi}_s\) from either the prior \(p(\bm{\phi}|\bm{\theta}_s)\) or the predictive distribution \(p(\bm{\phi}|\bm{\theta}_s, \bm{y}_1)\), to produce augmented samples \(\{(\bm{\theta}_s, \bm{\phi}_s)\}_{s=1}^S\) which are from \(p(\bm{\theta}, \bm{\phi}|\bm{y}_1)\). Then with these samples as proposals, proceed with PPRB as usual. \citet{taylor2023fast} note that this sample augmentation approach is problematic in situations where the prior \(p(\bm{\phi}|\bm{\theta}_s)\) or predictive distribution \(p(\bm{\phi}|\bm{\theta}_s, \bm{y}_1)\) is diffuse, because values drawn from these distributions are poor proposals for the full-conditional distribution \(p(\bm{\phi}|\bm{\theta}_s, \bm{y}_1, \bm{y}_2)\), leading to low acceptance rates of the independent MH proposals. This problem can arise, for example, if the prior \(p(\bm{\phi}|\bm{\theta}_s) = p(\bm{\phi})\) has no dependence on \(\bm{\theta}_s\). The authors propose PPRB-within-Gibbs as an adaptation to PPRB.

\begin{definition}
\protect\hypertarget{def:pprb-within-gibbs-2step}{}\label{def:pprb-within-gibbs-2step}

\textbf{PPRB-within-Gibbs}. Consider the streaming model defined in Eq. (\ref{eqn:streaming-model-basic}). Let there be existing posterior samples, \(\{\boldsymbol \theta^s\}_{s=1}^S\) from the distribution \(p(\boldsymbol \theta | \boldsymbol y_1)\). Then for the desired number of posterior samples,

\begin{enumerate}
\def\labelenumi{\arabic{enumi}.}
\tightlist
\item
  (Initialization) Initialize \(\boldsymbol \phi\) by drawing a value from its prior \(p(\boldsymbol \phi | \boldsymbol \theta_{s^\prime})\) using a randomly drawn \(s^\prime \sim \text{Uniform}(\{1, \dots, S\})\).
\item
  (PPRB step) Propose a new value \(\boldsymbol \theta^\ast\) by drawing from the existing posterior samples \(\{\boldsymbol \theta^s\}_{s=1}^S\) with replacement. Accept or reject the proposal using the MH ratio \begin{equation}\alpha = \min\left(\frac{p(\boldsymbol y_2 | \boldsymbol y_1, \boldsymbol \theta^\ast, \boldsymbol \phi)}{p(\boldsymbol y_2 | \boldsymbol y_1, \boldsymbol \theta, \boldsymbol \phi)}\frac{p(\boldsymbol \phi | \boldsymbol \theta^\ast)}{p(\boldsymbol \phi | \boldsymbol \theta)}, 1\right),\label{eqn:pprb-within-gibbs-accept-ratio}\end{equation}
\item
  Update the parameter \(\boldsymbol \phi\) from the full-conditional distribution \(p(\boldsymbol \phi | \boldsymbol \theta, \boldsymbol y_1, \boldsymbol y_2)\).
\end{enumerate}

\end{definition}

A formal specification of this algorithm is given as Algorithm \ref{alg:pprb-within-gibbs-2step} in the supplemental material. The PPRB step acceptance ratio (Eq. \ref{eqn:pprb-within-gibbs-accept-ratio}) is now the product of two ratios: the ratio of the distribution of the new data, and the ratio of the conditional prior of the new parameters, \(\bm{\phi}\). If the prior of \(\bm{\phi}\) is not dependent on \(\bm{\theta}\), i.e., \(p(\bm{\phi}|\bm{\theta}) = p(\bm{\phi})\), then this second ratio cancels. Further, if the data \(\bm{y}_1\) and \(\bm{y}_2\) are conditionally independent, the first ratio does not depend on \(\bm{y}_1\). \citet{taylor2023fast} originally proposed a three-step variant with a separate update for components within \(\bm{\theta}\) with conjugate full-conditionals. The simpler two-step PPRB-within-Gibbs algorithm is similar to Metropolis-Hastings Importance Resampling \citep{berzuini1997dynamic}, which was developed to address the growing parameter space complexity of dynamic conditional independence models. Unless otherwise noted, we refer to the algorithm in Definition \ref{def:pprb-within-gibbs-2step} as PPRB-within-Gibbs.

\hypertarget{sec:pprb-degradation}{%
\section{Filtering Degradation}\label{sec:pprb-degradation}}

Degradation refers to the tendency for MCMC samples resulting from filtering methods to contain many repeated values due to rejected proposals. As a result, when using existing a finite number of samples resampled with replacement as proposals for the next update \citep[as suggested in][]{hooten2021making}, the number of unique values decreases for continuous distributions. Due to degradation, the performance of filtering methods suffer for the streaming data setting, where a potentially large number of Bayesian updates must be performed while resampling the same pool of samples. As the number of unique values decreases within a set number, \(S\), of posterior samples produced after each update, the ability of those \(S\) samples to approximate the posterior distribution degrades. We next give an intuitive explanation for the causes of filtering degradation, provide theoretical bounds for the error introduced, and demonstrate filtering degradation via simulation in Section \ref{sec:simulation-studies}, all for the PPRB algorithm.

Each application of PPRB to perform a streaming update introduces error in the approximation to the updated posterior. To see this, consider a simple model with parameters \(\bm{\theta}\) and data partitioned in time, \(\bm{y}_t\) for \(t=1,2,\dots\). A finite number of samples \(\{\bm{\theta}^{(1)}_s\}_{s=1}^S\) are drawn from \(p(\bm{\theta}|\bm{y}_1)\). At time \(t=2\), \(\bm{y}_2\) are available and the posterior \(p(\bm{\theta}|\bm{y}_1,\bm{y}_2)\) must be approximated. By resampling \(\{\bm{\theta}_s^{(1)}\}_{s=1}^S\) with replacement as PPRB proposals, the true proposal distribution is the empirical posterior distribution of these samples, which we call \(F_S^{(1)}(\cdot)\). Of course, if an infinite number of samples are available after the first application of PPRB, \(F_S^{(1)}(\cdot)\) converges uniformly to the true distribution. However, in an actual application of PPRB, only a finite number of samples are obtained, resulting in an approximation. During the first PPRB update at \(t=2\), the target distribution is thus
\begin{equation*}
p(\bm{y}_2|\bm{\theta},\bm{y}_1)F_S^{(1)}(\bm{\theta}) \approx p(\bm{y}_2|\bm{\theta},\bm{y}_1)p(\bm{\theta}|\bm{y}_1) \propto p(\bm{\theta}|\bm{y}_1,\bm{y}_2),
\end{equation*}
an approximation of the desired posterior at \(t=2\). \(S\) samples, \(\{\bm{\theta}^{(2)}_s\}_{s=1}^S\), are drawn from this approximate posterior, which yields another empirical posterior distribution \(F^{(2)}_S(\cdot)\). The second PPRB update at \(t=3\) targets an approximation of an approximation of the desired posterior at \(t=3\),
\begin{align*}
p(\bm{y}_3|\bm{\theta},\bm{y}_1,\bm{y}_2)F_S^{(2)}(\bm{\theta}) &\approx p(\bm{y}_3|\bm{\theta},\bm{y}_1,\bm{y}_2)p(\bm{y}_2|\bm{\theta},\bm{y}_1)F_S^{(1)}(\bm{\theta}) \\
&\approx p(\bm{y}_3|\bm{\theta},\bm{y}_1,\bm{y}_2)p(\bm{y}_2|\bm{\theta},\bm{y}_1)p(\bm{\theta}|\bm{y}_1) \propto p(\bm{\theta}|\bm{y}_1,\bm{y}_2,\bm{y}_3).
\end{align*}

Intuitively, as this process repeats in subsequent updates, an accumulation of approximation error occurs. Whatever acceptable level of approximation for using PPRB once is surpassed in using PPRB a second time or beyond.

\hypertarget{sec:bounds-on-pprb-approximation-error}{%
\subsection{Bounds on PPRB Approximation Error}\label{sec:bounds-on-pprb-approximation-error}}

We derive upper and lower bounds on the PPRB approximation error in a streaming setting where PPRB is applied sequentially by resampling the finite number of samples produced at a previous stage as proposals in the next stage. We start by defining notation,
\begin{align}
\pi_t &= p(\bm{\theta}|\bm{y}_1,\dots,\bm{y}_t), \ \text{true posterior at time $t$}, \label{eqn:inequality-terms-start} \\
A_t &\propto p(\bm{y}_t|\bm{\theta},\bm{y}_1,\dots,\bm{y}_{t-1})F_S^{(t-1)}(\bm{\theta}), \ \text{approximate PPRB posterior at time $t$}, \\
F_S^{(t)} &= \text{empirical distribution of $S$ samples drawn from $A_t$}. \label{eqn:inequality-terms-end}
\end{align}
The quantity of interest is then \(\left\|A_t - \pi_t\right\|\), where \(\|\cdot\|\) is a norm on probability distributions. The following upper and lower bounds exist for this quantity.

\begin{theorem}
\protect\hypertarget{thm:pprb-error-bounds}{}\label{thm:pprb-error-bounds}Let \(\pi_t\), \(A_t\), and \(F_S^{(t)}\) be defined as in Eq. (\ref{eqn:inequality-terms-start})-(\ref{eqn:inequality-terms-end}) and let \(\|\cdot\|\) be a norm on probability measures. Then,
\begin{align}
\left\|A_t - \pi_t\right\| &\leq \underbrace{\left\|A_{t-1} - \pi_{t-1}\right\|}_{(1)} + \underbrace{\left\|F_S^{(t-1)} - A_{t-1}\right\|}_{(2)} + \underbrace{\left\|\pi_t - \pi_{t-1}\right\|}_{(3)} + \underbrace{\left\|A_t - F_S^{(t-1)}\right\|}_{(4)} \label{eqn:pprb-error-upper-bound} \\
\left\|A_t - \pi_t\right\| &\geq \bigg\lvert \underbrace{\left\|A_{t-1} - \pi_{t-1}\right\|}_{(1)} - \underbrace{\left\|F_S^{(t-1)} - A_{t-1}\right\|}_{(2)} \bigg\rvert - \underbrace{\left\|\pi_t - \pi_{t-1}\right\|}_{(3)} - \underbrace{\left\|A_t - F_S^{(t-1)}\right\|}_{(4)} \label{eqn:pprb-error-lower-bound}
\end{align}
\end{theorem}

\begin{proof}
Appendix \ref{sec:supplement-theorems}
\end{proof}

The form of these inequalities has the advantage that all terms on the right hand sides of Eq. (\ref{eqn:pprb-error-upper-bound}) and Eq. (\ref{eqn:pprb-error-lower-bound}) are easily interpreted in the context of streaming Bayesian updates:

\begin{enumerate}
\def\labelenumi{(\arabic{enumi})}
\tightlist
\item
  \(\left\|A_{t-1} - \pi_{t-1}\right\|\) is the existing approximation error at time \(t-1\).
\item
  \(\left\|F_S^{(t-1)} - A_{t-1}\right\|\) is the error introduced by producing a finite number, \(S\), of MCMC samples from the approximate posterior.
\item
  \(\left\|\pi_t - \pi_{t-1}\right\|\) is the difference in the true posterior after receiving \(\bm{y}_t\).
\item
  \(\left\|A_t - F_S^{(t-1)}\right\|\) is the difference in approximate distribution , from the prior \(F_S^{(t-1)}\) to the posterior \(A_t\) after receiving \(\bm{y}_t\).
\end{enumerate}

Together, Eq. (\ref{eqn:pprb-error-upper-bound}) and (\ref{eqn:pprb-error-lower-bound}) give upper and lower bounds for the approximation error at time \(t\), \(\left\|A_t - \pi_t\right\|\), which are especially useful for large \(t\). As \(t\to\infty\), the proportion of all data (\(\bm{y}_1,\dots,\bm{y}_t\)) contained in \(\bm{y}_t\) becomes smaller, so we have the intuition that the Bayesian updates will be smaller, in the sense that the difference between the ``prior'' and posterior goes to zero. Thus, \(\left\|\pi_t - \pi_{t-1}\right\| \to 0\) and \(\left\|A_t - F_S^{(t-1)}\right\| \to 0\) are reasonable to assume as \(t \to \infty\). Then \(\left\|A_t - \pi_t\right\|\) is approximately bounded by
\begin{align}
\left\|A_t - \pi_t\right\| &\lessapprox \left\|A_{t-1} - \pi_{t-1}\right\| + \left\|F_S^{(t-1)} - A_{t-1}\right\| \label{eqn:pprb-error-upper-bound-approx} \\
\left\|A_t - \pi_t\right\| &\gtrapprox \bigg\lvert \left\|A_{t-1} - \pi_{t-1}\right\| - \left\|F_S^{(t-1)} - A_{t-1}\right\| \bigg\rvert.  \label{eqn:pprb-error-lower-bound-approx}
\end{align}
For large \(t\), the inequalities in Eq. (\ref{eqn:pprb-error-upper-bound-approx}) and Eq. (\ref{eqn:pprb-error-lower-bound-approx}) show that the PPRB error at \(t\) is approximately bounded by a triangle inequality with \(\left\|A_{t-1} - \pi_{t-1}\right\|\), the approximation error at \(t-1\), and \(\left\|F_S^{(t-1)} - A_{t-1}\right\|\), the finite sample error at \(t-1\). If there exists some \(\epsilon > 0\) such that \(\left\|F_S^{(t)} - A_{t}\right\| > \epsilon\) for all \(t\), then the approximation error cannot decay to zero as \(t \to \infty\). For large \(t\), if \(\left\|A_{t} - \pi_{t}\right\| < \epsilon/2\) then \(\left\|A_{t+1} - \pi_{t+1}\right\| > \epsilon/2\). Similar results apply to PPRB-within-Gibbs, see Appendix \ref{sec:pprb-within-gibbs-approximation-error}.

\hypertarget{sec:generative-filtering}{%
\section{Generative Filtering}\label{sec:generative-filtering}}

We introduce Generative Filtering, an MCMC sampler for recursive Bayesian updates which avoids filtering degeneracy and achieves faster convergence than SMCMC. Consider data, \(\bm{y}\), partitioned into sequential batches, \(\bm{y}_t\) for \(t=1,2, \dots\), and data model
\begin{align*}
\bm{y}_1 &\sim p(\bm{y}_1 | \bm{\theta}_1), \\
\bm{y}_t &\sim p(\bm{y}_t | \bm{\theta}_1, \dots, \bm{\theta}_t, \bm{y}_1, \dots, \bm{y}_{t-1}), \text{ for $t \geq 2$}, \\
\bm{\theta}_1 &\sim p(\bm{\theta}_1), \quad
\bm{\theta}_t \sim p(\bm{\theta}_t | \bm{\theta}_1, \dots, \bm{\theta}_{t-1}), \text{ for $t \geq 2$}.
\end{align*}
That is, each batch of data has a distribution which can depend on some set of parameters and previously arrived batches of data. The parameters, \(\bm{\theta}\), are divided into batches, \(\bm{\theta}_t\) for \(t=1,2,\dots\), such that the batch \(\bm{y}_t\) of data depends only on the parameters \(\bm{\theta}_1, \dots, \bm{\theta}_t\). Each batch of parameters has a prior that may depend on parameters in previous batches. For convenience, define \(\bm{y}_{1:t} := (\bm{y}_1 \ \dots \ \bm{y}_t)\) and \(\bm{\theta}_{1:t} := (\bm{\theta}_1 \ \dots \ \bm{\theta}_t)\). The Generative Filtering algorithm begins at time \(t\), after the arrival of \(\bm{y}_t\).

\begin{definition}
\protect\hypertarget{def:generative-filtering}{}\label{def:generative-filtering}

\textbf{Generative Filtering}.

Let there be samples \(\{\bm{\theta}_{1:(t-1),s}\}_{s=1}^S\) from the posterior distribution \(p(\bm{\theta}_{1:(t-1)} | \bm{y}_{1:(t-1)})\). Let \(T_t\) be a transition kernel targeting the updated posterior distribution, \(p(\bm{\theta}_{1:t} | \bm{y}_{1:t})\). The Generative Filtering update consists of two steps:

\begin{enumerate}
\def\labelenumi{\arabic{enumi}.}
\tightlist
\item
  Perform a filtering step to produce a sample of size \(S\), \(\{\bm{\theta}^{\ast}_{1:t,s}\}_{s=1}^S\).
\item
  Using each \(\{\bm{\theta}^{\ast}_{1:t,s}\}_{s=1}^S\) as an initial value, apply the transition kernel \(T_t\) in parallel \(m_t\) times in \(S\) parallel chains, saving the final value of each chain.
\end{enumerate}

\end{definition}

A formal specification of this algorithm is given as Algorithm \ref{alg:generative-filtering} in the supplemental material. The resulting samples, \(\{\bm{\theta}_{1:t,s}\}_{s=1}^S\) approximate the updated posterior distribution, \(p(\bm{\theta}_{1:t} | \bm{y}_{1:t})\). The filtering method in step 1 is a method that resamples the ensemble \(\{\bm{\theta}_{1:(t-1),s}\}_{s=1}^S\) to produce samples approximating the distribution, \(p(\bm{\theta}_{1:t} | \bm{y}_{1:t})\). When the filtering method in step 1 is importance resampling \citep{rubin1988using} or the weighted bootstrap \citep{smith1992bayesian}, Generative Filtering encompasses the resample-move algorithm of \citet{gilks2001following}. When the filtering method is instead a PPRB-within-Gibbs update, the Generative Filtering algorithm can be seen as an extension of both PPRB-within-Gibbs and SMCMC. It extends PPRB-within-Gibbs by applying some number of transition kernel steps in parallel to each resulting sample and extends SMCMC by replacing the jumping kernel with a PPRB-within-Gibbs update. The transition kernel application requires for each initialization value \(\bm{\theta}_{1:t}^{(0)} = \bm{\theta}^*_{1:t,s}\) for \(i = 1, \dots, m_t\)
a sample is drawn \(\bm{\theta}_{1:t, s}^{(i)}\) from the density \(p\left(\bm{\theta}_{1:t}|\bm{y}_{1:t}, \bm{\theta}_{1:t,s}^{(i - 1)}\right) = T_t\left(\bm{\theta}_{1:t,s}^{(i - 1)}, \bm{\theta}_{1:t}\right)\). The transition kernel \(T_t\) is defined so that the posterior \(p(\bm{\theta}_{1:t}|\bm{y}_{1:t})\) is the stationary measure of the Markov chain with transition kernel \(T_t\). After \(m_t\) iterations, the final value of each of the \(S\) parallel chains, \(\left\{\bm{\theta}_{1:t, s}^{(m_t)}\right\}_{s = 1}^S\), is saved. The use of a filtering step rather than a jumping kernel allows for the update of all parameters in the ensemble instead of just new parameters. Generative Filtering seeks to quickly converge to and sample from its target distribution through the use of an existing posterior sample from an earlier posterior distribution, while avoiding the problem of filtering methods that degrade to degenerate distributions.

\hypertarget{convergence-results}{%
\subsection{Convergence Results}\label{convergence-results}}

We derive bounds for the convergence of Generative Filtering to its target posterior distribution, and find sufficient conditions under which Generative Filtering has a stronger convergence guarantee than SMCMC in fewer iterations of the transition kernel.

\begin{theorem}
\protect\hypertarget{thm:theorem39}{}\label{thm:theorem39}Let \(P^S_t(\bm{\theta}_{1:(t-1)},\cdot)\) represent the kernel resulting from \(S\) applications of a filtering method at time \(t\), which is a probability density for \(\bm{\theta}_{1:t} := (\bm{\theta}_{1:(t-1)}, \bm{\theta}_t)\). Let \(\pi_t = p(\bm{\theta}_{1:t}|\bm{y}_{1:t})\) be the target posterior at time \(t\). Assuming the following conditions:

\begin{enumerate}
\def\labelenumi{\arabic{enumi}.}
\tightlist
\item
  (Universal ergodicity) There exist \(\rho_t \in (0,1)\), such that for all \(t>0\) and \(x \in {\cal X}\), \[||T_t(x, \cdot) - \pi_t||_1 \leq 2\rho_t.\]
\item
  (Filtering consistency) For a sequence of \(\lambda_t \to 0\) and a bounded sequence of positive integers \(S_t\), the following holds: \[\underset{\bm{\theta}_{1:(t-1)}}{\sup}||\pi_t - P^{S_t}_t(\bm{\theta}_{1:(t-1)},\cdot)||_1 \leq 2\lambda_t.\]
\end{enumerate}

Let \(\epsilon_t = \rho_t^{m_t}\) and let \(Q_t = T_t^{m_t} \circ P_t^{S_t}\) be a Generative Filtering update at time \(t\). Then for any initial distribution \(\pi_0\),
\begin{equation*}||Q_t \circ \cdots \circ Q_1 \circ \pi_0 - \pi_t||_1 \leq \sum_{v=1}^t\left\{\prod_{u=v+1}^t \epsilon_u(1-\lambda_u)\right\}\epsilon_v\lambda_v \leq \sum_{v=1}^t\left\{\prod_{u=v}^t \epsilon_u\right\}\lambda_v.\end{equation*}
\end{theorem}

\begin{remark}
When PPRB-within-Gibbs is used as the filtering method within Generative Filtering, the Filtering consistency condition contains two key requirements related to PPRB-within-Gibbs and the posterior distributions, \(\pi_1,\dots,\pi_t\). First, the PPRB-within-Gibbs transition kernel must produce an irreducible Markov chain. A sufficient condition for irreducibility is the positivity criterion for PPRB-within-Gibbs from \citet{taylor2023fast}. Second, in order for the sequence of \(S_t\) to be bounded, the difference between subsequent posteriors must not be too large as \(t\) increases, or the PPRB-within-Gibbs transition kernel must have fast enough convergence to overcome the differences.
\end{remark}

\begin{proof}
Appendix \ref{sec:supplement-theorems}
\end{proof}

\hypertarget{sec:comparison-bounds}{%
\subsubsection{Comparison of upper bounds}\label{sec:comparison-bounds}}

We next compare the upper bounds for convergence of SMCMC and Generative Filtering, providing conditions for when Generative Filtering's upper bound is at least as small as that of SMCMC. We compare the SMCMC algorithm with jumping kernel, \(J_t\), and transition kernel, \(T_t\), to the Generative Filtering update composed of a filtering step, \(P_t^{S_t}\), and the same transition kernel, \(T_t\).

\begin{theorem}
\protect\hypertarget{thm:bound-inequality}{}\label{thm:bound-inequality}Assume the following conditions hold:

\begin{enumerate}
\def\labelenumi{\arabic{enumi}.}
\tightlist
\item
  (Universal ergodicity) There exists \(\rho_t \in (0,1)\), such that for all \(t>0\) and \(x \in {\cal X}\), \[||T_t(x, \cdot) - \pi_t||_1 \leq 2\rho_t.\]
\item
  (Stationary convergence) The stationary distribution \(\pi_t\) of \(T_t\) satisfies \[\alpha_t = \frac{1}{2}||\pi_t - \pi_{t-1}||_1 \to 0,\] where \(\pi_t\) is the marginal posterior of \(\theta_{1:(t-1)}\) at time \(t\) in \(\alpha_t\).
\item
  (Filtering consistency) For a sequence of \(\lambda_t^{(F)} \to 0\) and a bounded sequence of positive integers \(S_t\), the following holds: \[\underset{\theta_{1:(t-1)}}{\sup}||\pi_t - P^{S_t}_t(\theta_{1:(t-1)},\cdot)||_1 \leq 2\lambda_t^{(F)}.\]
\item
  (Jumping consistency) For a sequence of \(\lambda_t^{(J)} \to 0\), the following holds: \[\underset{\theta_{1:(t-1)}}{\sup} ||\pi_t(\cdot | \theta_{1:(t-1)}) - J_t(\theta_{1:(t-1)}, \cdot)||_1 \leq 2\lambda_t^{(J)}.\]
\end{enumerate}

Let \(\epsilon_t = \rho_t^{m_t}\). Define \[\gamma^{(F)}_t = \sum_{v=1}^t\left\{\prod_{u=v+1}^t \epsilon_u(1-\lambda^{(F)}_u)\right\}\epsilon_v\lambda^{(F)}_v\] and \[\gamma^{(J)}_t = \sum_{v=1}^t\left\{\prod_{u=v}^t \epsilon_u\right\}(\lambda^{(J)}_v + \alpha_v)\] to be the bounds from Theorem \ref{thm:theorem39} and Theorem 3.9 of \citet{yang2013sequential}, respectively. If, for all \(u \leq t\), \(\lambda^{(F)}_u \leq \alpha_u + \lambda^{(J)}_u\), then \(\gamma_t^{(F)} \leq \gamma_t^{(J)}\).
\end{theorem}

\begin{proof}
Appendix \ref{sec:supplement-theorems}
\end{proof}

Theorem \ref{thm:bound-inequality} gives a sufficient, but somewhat restrictive, condition on the convergence of relative pieces such that Generative Filtering's convergence is bounded more tightly than SMCMC. It reveals an interesting relationship between the use of filtering or the jumping kernel. For a fixed \(t \geq 1\), no matter how good the jumping kernel is, e.g., \(\lambda^{(J)}_t = 0\), there is a fixed \(\alpha_t > 0\) contributing to the SMCMC upper bound. When PPRB-within-Gibbs is used as the filtering method, there exists an \(S_t\) such that \(\lambda^{(F)}_t \leq \alpha_t \leq \alpha_t + \lambda^{(J)}_t\).

\begin{theorem}
\protect\hypertarget{thm:bound-recursive-inequality}{}\label{thm:bound-recursive-inequality}With the conditions and definitions of Theorem \ref{thm:bound-inequality}, assume \(\gamma^{(F)}_{t-1} = \gamma^{(J)}_{t-1}\) and define \(\gamma := \gamma^{(F)}_{t-1} = \gamma^{(J)}_{t-1}\). If \(\gamma < 1\) and \(\lambda^{(F)}_t \leq \frac{\alpha_t + \lambda^{(J)}_t}{1-\gamma}\), then \(\gamma^{(F)}_{t} \leq \gamma^{(J)}_{t}\). If \(\gamma \geq 1\) then \(\gamma^{(F)}_{t} \leq \gamma^{(J)}_{t}\) always.
\end{theorem}

\begin{proof}
Appendix \ref{sec:supplement-theorems}
\end{proof}

Theorem \ref{thm:bound-recursive-inequality} shows a sufficient condition for a claim that is slightly weaker than that in Theorem \ref{thm:bound-inequality}, but as a result, the condition is less strict. Theorem \ref{thm:bound-recursive-inequality} is relevant to the relative gain in the bound of Generative Filtering over that of SMCMC, when both are starting from the same position. As a result we only need conditions on time \(t\), not all times \(u \leq t\). In the \(\gamma < 1\) case, this condition is less strict than the equivalent condition in Theorem \ref{thm:bound-inequality}, because of the denominator \(0 < 1 - \gamma \leq 1\). For larger \(\gamma\), the bound on the filtering step can be significantly worse than the bound on the jumping kernel and still result in a lower upper bound for the total Generative Filtering process. In the \(\gamma \geq 1\) case, when starting from the same condition, the upper bound for Generative Filtering at the next time point will \emph{always} be lower than that of SMCMC.

\hypertarget{transition-kernel-and-m_t-the-required-iterations}{%
\subsection{\texorpdfstring{Transition kernel and \(m_t\), the required iterations}{Transition kernel and m\_t, the required iterations}}\label{transition-kernel-and-m_t-the-required-iterations}}

In the above Section \ref{sec:comparison-bounds}, we assume that Generative Filtering and SMCMC use the same transition kernel for the same number of iterations, \(m_t\). In this section, we describe under what conditions Generative Filtering can use fewer transition kernel iterations to achieve the same convergence. As in the previous section, we compare the SMCMC algorithm with jumping kernel, \(J_t\), and transition kernel, \(T_t\), to the Generative Filtering update composed of a filtering step, \(P_t^{S_t}\), and the same transition kernel, \(T_t\).

Lemmas 3.2 and 3.1 of \citet{yang2013sequential} together show that if a transition kernel \(T_t\) satisfies \(\underset{x}{\sup}||T_t(x, \cdot) - \pi_t||_1 \leq 2\rho_t\), then for a distribution \(p_0\) we have \(||T_t^{m_t} \circ p_0 - \pi_t||_1 \leq \rho_t^{m_t}||p_0 - \pi_t||_1\). We wish to compare the case when \(p_0 = J_t \circ \pi_{t-1}\) and when \(p_0 = P^{S_t}_t \circ \pi_{t-1}\). Thus, filtering results in a distribution closer to \(\pi_t\) than the jumping kernel, the result of Generative Filtering using the same transition kernel with the same number of steps as SMCMC achieves a smaller upper bound on its distance from \(\pi_t\) than SMCMC.

Now, consider the case when the transition kernel is applied a different number of times in Generative Filtering (\(m_t^{(F)}\)) than in SMCMC (\(m_t^{(S)}\)). If
\begin{equation}
m_t^{(F)} \geq m_t^{(S)} - \frac{\log\left(||J_t \circ \pi_{t-1} - \pi_t||_1 / || P^{S_t}_t \circ \pi_{t-1} - \pi_t||_1\right)}{\log(1/\rho_t)}, \label{eqn:transition-kernel-step-inequality}
\end{equation}
then \(\rho_t^{m^{(F)}_t}||P^{S_t}_t \circ \pi_{t-1} - \pi_t||_1 \leq \rho_t^{m^{(S)}_t}||J_t \circ \pi_{t-1} - \pi_t||_1\). Since \(\rho_t < 1\), \(\log(1/\rho_t) > 0\) and if \(||P^{S_t}_t \circ \pi_{t-1} - \pi_t||_1 < ||J_t \circ \pi_{t-1} - \pi_t||_1\), then
\begin{equation*}
\frac{\log\left(||J_t \circ \pi_{t-1} - \pi_t||_1 / || P^{S_t}_t \circ \pi_{t-1} - \pi_t||_1\right)}{\log(1/\rho_t)} > 0.
\end{equation*}
Generative Filtering can achieve a lower convergence bound than SMCMC with either more transition kernel steps, or some smaller number of steps determined by Eq. (\ref{eqn:transition-kernel-step-inequality}).

In order for fewer transition kernel steps to be required, we need
\begin{equation*}
\frac{\log\left(||J_t \circ \pi_{t-1} - \pi_t||_1 / || P^{S_t}_t \circ \pi_{t-1} - \pi_t||_1\right)}{\log(1/\rho_t)} \geq k,
\end{equation*}
where the integer \(k \geq 1\) is the number of steps difference, which is true if and only if
\begin{equation}
\frac{|| P^{S_t}_t \circ \pi_{t-1} - \pi_t||_1}{||J_t \circ \pi_{t-1} - \pi_t||_1} \leq \rho_t^k. \label{eqn:transition-kernel-step-condition}
\end{equation}
This reveals a relationship between the gains of PPRB-within-Gibbs over the jumping kernel, \(|| P^{S_t}_t \circ \pi_{t-1} - \pi_t||_1/||J_t \circ \pi_{t-1} - \pi_t||_1\), and the mixing of the transition kernel, \(\rho_t\). If the transition kernel is weakly mixing, \(\rho_t \approx 1\), PPRB-within-Gibbs can result in a larger reduction in the number of required transition kernel steps over the SMCMC jumping kernel than in the case when the transition kernel is strongly mixing, \(\rho_t \ll 1\).

While Eq. (\ref{eqn:transition-kernel-step-condition}) provides a sufficient condition for Generative Filtering to require fewer transition kernel steps for an equivalent convergence bound to SMCMC, a lower number of transition kernel steps does not necessarily mean Generative Filtering takes less time to complete. This is because while SMCMC's jumping kernel can be applied to each ensemble sample in parallel, filtering methods such as PPRB-within-Gibbs are inherently sequential. Thus when many cores are available for computation, the filtering step of Generative Filtering is a potential bottleneck. The improvement in convergence provided by PPRB-within-Gibbs may or may not overcome the time saved by parallelizing the jumping kernel in SMCMC. In Section \ref{sec:simulation-studies} and Section \ref{sec:application}, we demonstrate both cases.

\hypertarget{streaming-data-storage-considerations}{%
\subsection{Streaming Data Storage Considerations}\label{streaming-data-storage-considerations}}

In this section, we provide conditions under which a transition kernel at time \(t\) can be applied without requiring storage of the full data through time \(t\). In general, we want a sufficient statistic, \(U_t(\bm{y}_t)\), to exist for each batch of data, \(\bm{y}_1, \bm{y}_2, \dots\), so that the likelihood can be evaluated without the full data, and transition kernels can be applied while needing to store only the statistics \(U_t\) for \(t=1,2,\dots\). The data themselves are always a trivial sufficient statistic, however when \(\dim U_t \ll \dim \bm{y}_t\), storage of the sufficient statistics allows for significant reductions in data storage requirements and potentially faster computation of the transition kernel. We describe conditions that are sufficient to allow reduced data storage when applying the transition kernel.

\begin{theorem}
\protect\hypertarget{thm:data-exponential-family}{}\label{thm:data-exponential-family}Assume:

\begin{enumerate}
\def\labelenumi{\arabic{enumi}.}
\tightlist
\item
  The data \(\bm{y}_{t_1}\) and \(\bm{y}_{t_2}\), for all \(t_1 < t_2\), are conditionally independent given \(\bm{\theta}_{1:t_2}\).
\item
  Each \(\bm{y}_t\) is a sample of \(n_t\) i.i.d. observations \(\bm{y}_{t,i}\) for \(i=1,\dots,n_t\).
\item
  Each observation \(\bm{y}_{t,i}\) comes from an exponential family distribution.
\end{enumerate}

Then storage of the full data can be avoided through the use of sufficient statistics.
\end{theorem}

\begin{proof}
Appendix \ref{sec:supplement-theorems}
\end{proof}

Note that in Theorem \ref{thm:data-exponential-family}, there is no requirement that the data distribution of \(\bm{y}_t\) be the same for each time \(t\), nor that the sufficient statistics \(U_t(\cdot)\) be the same function at each time \(t\). These conditions are thus flexible, applying to a wide range of streaming data settings including time-varying data sources, data collection techniques, or measurement devices that result in changing data distributions. Additionally, the Fisher-Darmois-Koopman-Pitman Theorem \citep{daum1986fisher} holds that any distribution with a smooth nowhere-vanishing density and finite-dimensional sufficient statistics is necessarily an exponential family distribution. Therefore, the exponential family encompasses most relevant distributions with sufficient statistics.

\hypertarget{sec:simulation-studies}{%
\section{Simulation Studies}\label{sec:simulation-studies}}

In this section, we apply Generative Filtering to two sets of simulated data, a state-space model and a streaming record linkage model, and compare to other streaming samplers for Bayesian updates with respect to speed and sampling accuracy.

\hypertarget{sec:simulation-gaussian}{%
\subsection{State-Space Model}\label{sec:simulation-gaussian}}

We analyze the following state-space model,
\begin{align*}
y_{t,i} &\sim \text{N}(\theta_t, \sigma^2), \text{ for $t=1,2,\dots$ and $i=1,\dots,n$} \\
\theta_1 &\sim \text{N}(0, \phi^2), \quad \theta_t | \bm \theta_{1:(t - 1)} \sim \text{N}(\theta_{t-1}, \phi^2), \text{ for $t>1$},
\end{align*}
where \(n\), \(\sigma^2\), and \(\phi^2\) are fixed. This model is a simple representation of a streaming data problem where a new state variable, \(\theta_t\), is required to parameterize the distribution of data \(\bm{y}_t\). As the data arrive sequentially, \(\bm{y}_1, \bm{y}_2, \dots\), updated estimates of the state variables, \(\bm{\theta}_{1:t}\), are desired through the posterior distribution of \(p(\bm{\theta}_{1:t}|\bm{y}_{1:t})\).

Data were generated at each combination of \(n=1,5,10,50\) and \(\sigma^2=0.25, 0.5, 1, 2, 4\), with \(\phi^2=1\). For practical reasons of generating the data, we generate data for \(t=1,\dots,T\) with \(T=100\). The values of \(n\) and \(\sigma^2\) create varying levels of signal from the data arriving at each time. For each of the 20 combinations of \(n\) and \(\sigma^2\), 20 sets of data were produced with differing random seeds by first drawing \(\bm{\theta}_{1:T}\) from \(\theta_1 \sim \text{N}(0, \phi^2)\) and \(\theta_t|\theta_{t-1} \sim N(\theta_{t-1},\phi^2)\) for \(2 \leq t \leq T\), then drawing \(\bm{y}_{1:T}\) from \(\bm{y}_t|\theta_t \overset{\text{i.i.d.}}{\sim} \text{N}(\theta_t, \sigma^2)\) for \(1 \leq t \leq T\).

\begin{figure}
\centering
\includegraphics{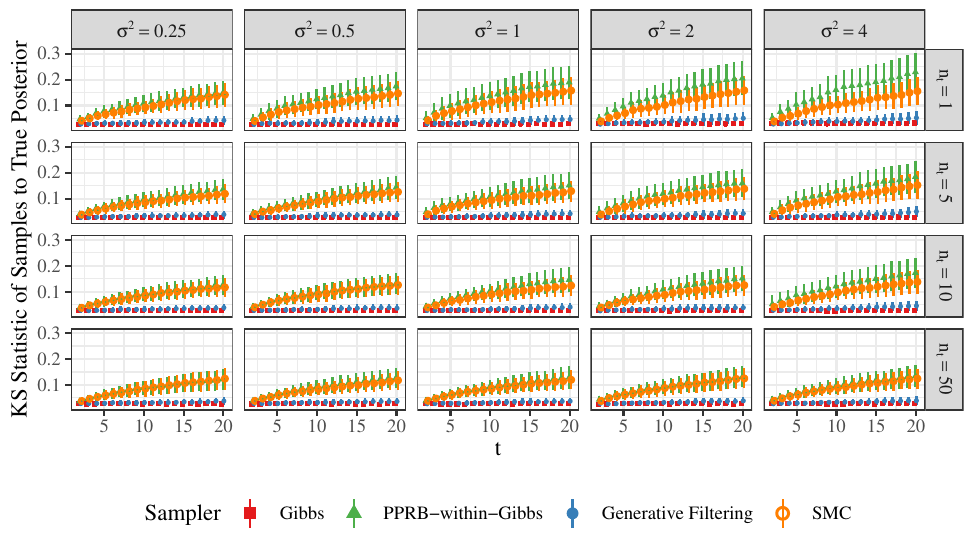}
\caption{\label{fig:simulation-plot-degrading}KS statistic (mean plus or minus standard deviation) for samples of \(\theta_1\) after repeated application of non-streaming Gibbs, PPRB-within-Gibbs, SMC, and Generative Filtering, for \(t=2,\dots,20\). Higher KS statistics indicate more MCMC error in the samples. After repeated applications, PPRB-within-Gibbs and SMC produce worse approximations to the posterior distribution while Generative Filtering does not.}
\end{figure}

For each simulated dataset, and for each value of \(t=1,\dots,T\), we estimated the posterior distribution of \(p(\bm{\theta}_{1:t}|\bm{y}_{1:t})\) first using four samplers: non-streaming Gibbs, PPRB-within-Gibbs, SMC, and Generative Filtering. Each component, \(\theta_\ell\), in \(\bm{\theta}_{1:t}\) has a conjugate Gaussian full-conditional distribution, making a Gibbs sampler a natural choice for a non-streaming MCMC. The Gibbs sampler produced samples directly from the posterior distribution, \(p(\bm{\theta}_{1:t}|\bm{y}_{1:t})\), for \(t=2,\dots,T\). Each Gibbs update was run for 1100 iterations, discarding the first 100 as burn-in with 10 independent chains.

The PPRB-within-Gibbs, SMC and Generative Filtering updates were sequentially applied in a streaming fashion using samples from the previous time point as the initial ensemble. The first streaming updates at time \(t=2\) used samples drawn independently from the true Gaussian posterior at \(t=1\), \(p(\theta_1|y_1)\), as their initial ensemble. For the update of the new state variable, \(\theta_t\), in PPRB-within-Gibbs, we used its full-conditional Gaussian distribution, \(\theta_t|\theta_{t-1},\bm{y}_t \sim \text{N}\left(V_t C_t,V_t\right)\) with \(V_t = \left(1/\phi^2 + n_t/\sigma^2\right)^{-1}\) and \(C_t = \theta_{t-1}/\phi^2 + \sum_{i=1}^{n_t}y_{t,i}/\sigma^2\). Each PPRB-within-Gibbs sampler was run for 1100 iterations, discarding the first 100 as burn-in. The PPRB-within-Gibbs updates that served as the filtering step in Generative Filtering used the same configuration. The transition kernel used in Generative Filtering updates all state variables \(\bm{\theta}_{1:t}\) simultaneously using a Metropolis random walk with multivariate Gaussian proposals. As the true posterior distribution is known, we choose a proposal distribution based on the adaptive Metropolis proposal of \citet{gelman1997weak} and \citet{haario2001adaptive}, i.e.~\(\bm{\theta}_{1:t}^\ast | \bm{\theta}_{1:t} \sim \text{N}(\bm{\theta}_{1:t}, 2.4^2\bm{\Sigma}/t)\), where \(\bm{\Sigma}\) is the true posterior covariance matrix of \(p(\bm{\theta}_{1:t}|\bm{y}_{1:t})\). We use a random walk Metropolis transition kernel to simulate a slowly-converging transition kernel, specifically, one which converges more slowly as the number of variables increases. This slow convergence can be an issue in many real-world MCMC applications. The transition kernel in Generative Filtering was applied \(m_t=5\) times for each update. For the SMC update we use the Bootstrap filter algorithm given in \citet{kantas2009overview} with the full-conditional \(p(\theta_t|\bm{y}_t,\theta_{t-1})\) as the importance density. This algorithm was developed to estimate the parameters of a state-space model. The full streaming update process from time \(t=2\) to \(t=T\) was run 10 times on each dataset, using 10 different random seeds.

To explore the phenomenon of filtering degradation in this simulated data, we compare the samples produced by Gibbs, PPRB-within-Gibbs, SMC and Generative Filtering to the true posterior distribution of \(p(\bm{\theta}_{1:t}|\bm{y}_{1:t})\). We used the max absolute difference between the true CDF and empirical CDF, i.e., the Kolmogorov-Smirnov test statistic, of each state variable, \(\theta_1, \dots, \theta_t\), as a measure of the MCMC error in the samples for that variable. Figure \ref{fig:simulation-plot-degrading} shows these measured errors for \(\theta_1\) at each time up to \(t=20\). As expected, the non-streaming Gibbs sampler is able to maintain a constant level of MCMC error at each time \(t=2,\dots,20\) because it is producing new samples from the posterior at each time. However, the MCMC error of PPRB-within-Gibbs and SMC increases with each update as the number of unique values of \(\theta_1\) in its samples decreases. Generative Filtering avoids this degradation, having MCMC error comparable to Gibbs.

\begin{figure}
\centering
\includegraphics{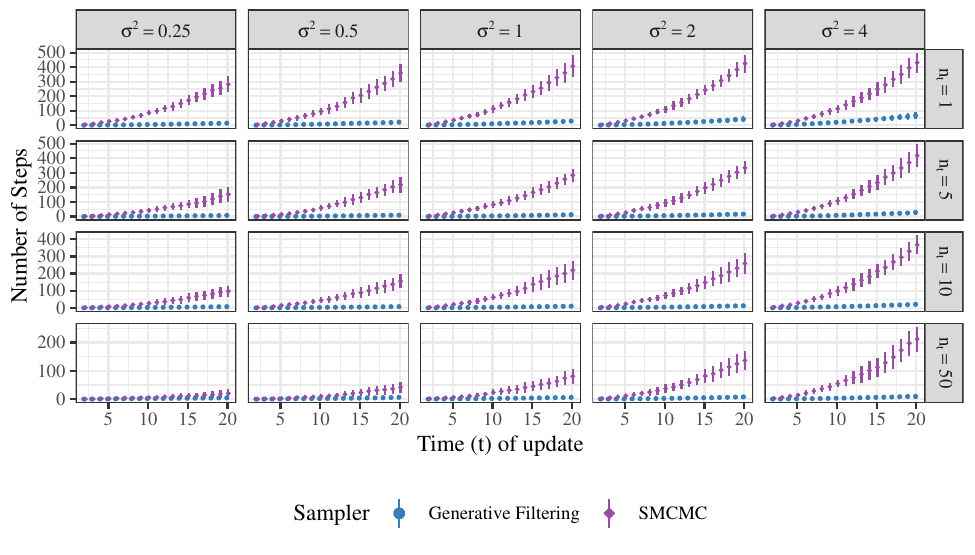}
\caption{\label{fig:simulation-plot-steps}Cumulative number (mean plus or minus standard deviation) of transition kernel steps to converge to the posterior distribution.}
\end{figure}

To investigate the convergence of Generative Filtering and SMCMC, we also compared the required number of transition kernel steps for SMCMC and Generative Filtering to estimate the state variables \(\bm{\theta}_{1:t}\) using the simulated data. At each time, \(t\), the jumping kernel for SMCMC and the update for \(\theta_t\) in the PPRB-within-Gibbs step of Generative Filtering were chosen to be the conjugate Gaussian full-conditional distribution, \(p(\theta_t|\theta_{t-1},\bm{y}_t)\). The Metropolis random walk transition kernel was used in both SMCMC and Generative Filtering, and the PPRB-within-Gibbs filtering step of Generative filtering was as described earlier. At each time \(t\) we began both Generative Filtering and SMCMC using the same initial ensemble of 1000 samples from \(p(\bm{\theta}_{1:(t-1)}|\bm{y}_{1:(t-1)})\) produced by the non-streaming Gibbs sampler. Thus, we are investigating the difference in minimum transition steps to converge for each method in individual updates, rather than over the course of all updates. This allows us to assess the effect of each method on convergence at every step.

We used an oracle stopping criteria to determine \(m_t\). Each sampler was stopped when the KS-statistics comparing the marginal empirical posterior distributions of \(\theta_t\) and \(\theta_{t-1}\) in the current ensemble to their respective true posterior distributions were both below 0.055, the critical value of the KS distribution. The cumulative number of transition kernel steps required to converge for each sampler up to each time is shown in Figure \ref{fig:simulation-plot-steps}. By \(t=20\), Generative Filtering has required many fewer transition kernel steps than SMCMC in all datasets due to its use of PPRB-within-Gibbs instead of the jumping kernel. However, while the jumping kernel in SMCMC is parallelizable, PPRB-within-Gibbs is inherently sequential. This presents a trade-off in environments with parallel computation available. We compare the cumulative runtime required for each method to converge in Appendix \ref{sec:appendix-simulation-study}. In this numerical experiment, for SMCMC to be faster than Generative Filtering on average in all settings would require 45 cores. A comparison to the correlation-based criteria to choose \(m_t\) of \citet{yang2013sequential} is in Appendix \ref{sec:appendix-simulation-study}.

\hypertarget{sec:simulation-rl}{%
\subsection{Streaming Record Linkage}\label{sec:simulation-rl}}

Record linkage is the task of consolidating records that refer to overlapping sets of entities. Often this task must be performed without the presence of a unique identifier. Noisy duplicated records present a problem for those who wish to use the data to make inferences. Bayesian record linkage estimates duplicate entities via a posterior distribution of linkages and provides natural uncertainty quantification \citep[e.g.,][]{tancredi2011hierarchical, steorts2016bayesian, sadinle2017bayesian, zanella2020informed, aleshinguendel2021multifile}.

\begin{figure}
\centering
\includegraphics{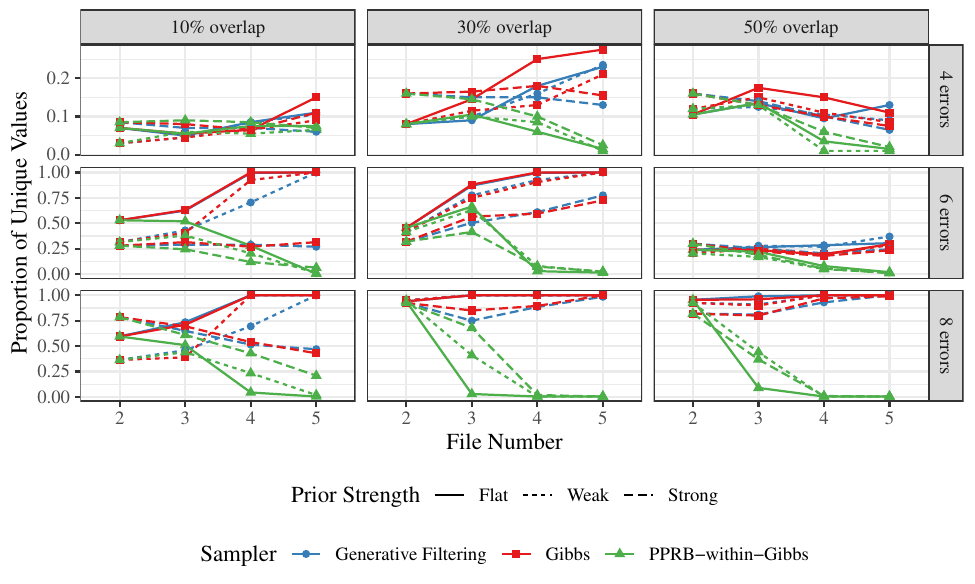}
\caption{\label{fig:streaming-rl-unique}Proportion of unique values within samples of \(Z^{(1)}\) produced by Gibbs, PPRB-within-Gibbs, and Generative Filtering. While the proportion of unique values from the PPRB-within-Gibbs samples decreases with each successive update, the proportion of unique values from Generative Filtering remains similar to the samples produced by Gibbs.}
\end{figure}

In streaming record linkage, sets of records (files) arrive sequentially in time with no predetermined number and estimates of links are updated after the arrival of each. Recent advancements in record linkage allow for near real-time data-driven record linkage \citep[e.g.,][]{christen2009similarity, ioannou2010onthefly, dey2011online, altwaijry2017qda, karapiperis2018summarization} and model-driven Bayesian record linkage in the streaming data setting can be performed using recursive Bayesian updates with SMCMC or PPRB-within-Gibbs \citep{taylor2023fast}. \citet{taylor2023fast} note PPRB-within-Gibbs is efficient, but degrades, and SMCMC requires a parallel computing environment but maintains accuracy. We expand on the simulation study performed in that paper by comparing to Generative Filtering. We demonstrate that Generative Filtering does not suffer from the same degradation as PPRB-within-Gibbs on the streaming record linkage problem.

For this study, we use the data files as described in \citet{taylor2023fast}. Data were simulated using the GeCo software package \citep{tran2013geco} to create realistic demographic information and insert realistic errors into randomly selected fields. Each record had 10 fields: given name, surname, age, occupation, and 6 categorical fields with 12 possible levels. Each record had errors inserted to mimic common typos, misspellings, or OCR errors. Files were generated with records that contained up to 4, 6, or 8 errors. The files were generated with varying overlap, having 10\%, 30\%, 50\% or 90\% of their records coreferent with records with previous files. For each of the 12 combinations of error and overlap, a set of 4 files was generated to arrive sequentially, simulating streaming data.

We use the streaming record linkage model of \citet{taylor2023fast}, with diffuse, weak, and strong priors for the parameter governing the distribution of disagreement levels for matched records. Details on the model formulation are included in Appendix \ref{sec:streaming-rl-model-details}. We begin the streaming updates using posterior samples from a two file record linkage using the \texttt{BRL} package \citep{sadinle2017bayesian}. Then, we perform sequential PPRB-within-Gibbs updates and sequential Generative Filtering updates with the third, fourth, and fifth file of each dataset. For comparison, we also fit the model using a non-streaming Gibbs sampler. The initial Gibbs sampler for two-file linkage was run using component-wise link updates for 11,000 iterations, discarding the first 1,000 as burn-in. This initial run was used as a basis for all streaming methods. The Gibbs sampler for subsequent files was run using component-wise link updates for 2,500 iterations, discarding the first 500 as burn-in. The PPRB-within-Gibbs sampler was run using locally balanced link updates for 11,000 iterations, discarding the first 1,000 as burn-in. The Generative Filtering updates consisted of an ensemble of 200 samples. Each filtering step used PPRB-within-Gibbs for 11,000 iterations, discarding the first 1000 and thinning to 200 before applying the transition kernel. The Generative Filtering transition kernel used locally balanced proposals for link updates, and was run for 200 iterations. The number of transition kernel iterations was conservatively chosen by examining traceplots for convergence. The PPRB-within-Gibbs updates use the three step version originally proposed by \citet{taylor2023fast} and implemented in the \texttt{bstrl} R package \citep{pkgbstrl}.

We examine the unique values within the produced samples for \(\bm{Z}^{(1)}\), the parameter that encodes links from file 2 to file 1. Before comparing, we thinned samples produced by Gibbs and PPRB-within-Gibbs from 2,000 and 10,000, respectively, to 200 to match the ensemble size of Generative Filtering. The number of unique values produced by PPRB-within-Gibbs decreases with each Bayesian update relative to Gibbs, while the number of unique values produced by Generative Filtering is comparable to Gibbs (Figure \ref{fig:streaming-rl-unique}). This indicates that Generative Filtering is able to avoid the degradation problems of PPRB-within-Gibbs.

\hypertarget{sec:application}{%
\section{Application}\label{sec:application}}

We investigate the effectiveness of Generative Filtering on data of Steller sea lion pup counts. Steller sea lions are an endangered species\footnote{United States Endangered Species Act of 1973} whose population has been impacted due to a reduction in juvenile survival attributed to pollution, oceanographic changes, and the effects of fisheries, among other potential causes \citep{fritz1995threatened}. The National Marine Mammal Laboratory of the National Oceanographic and Atmospheric Administration performs aerial surveys of the number of pups born each year across several sites in Alaska. The data are available in the \texttt{R} package \texttt{agTrend} \citep{pkgagTrend} and contain 713 observations between the years 1973 and 2016, at 72 sites. Not every site is measured at every year, with sites having between 1 and 22 observations.

\begin{figure}
\centering
\includegraphics{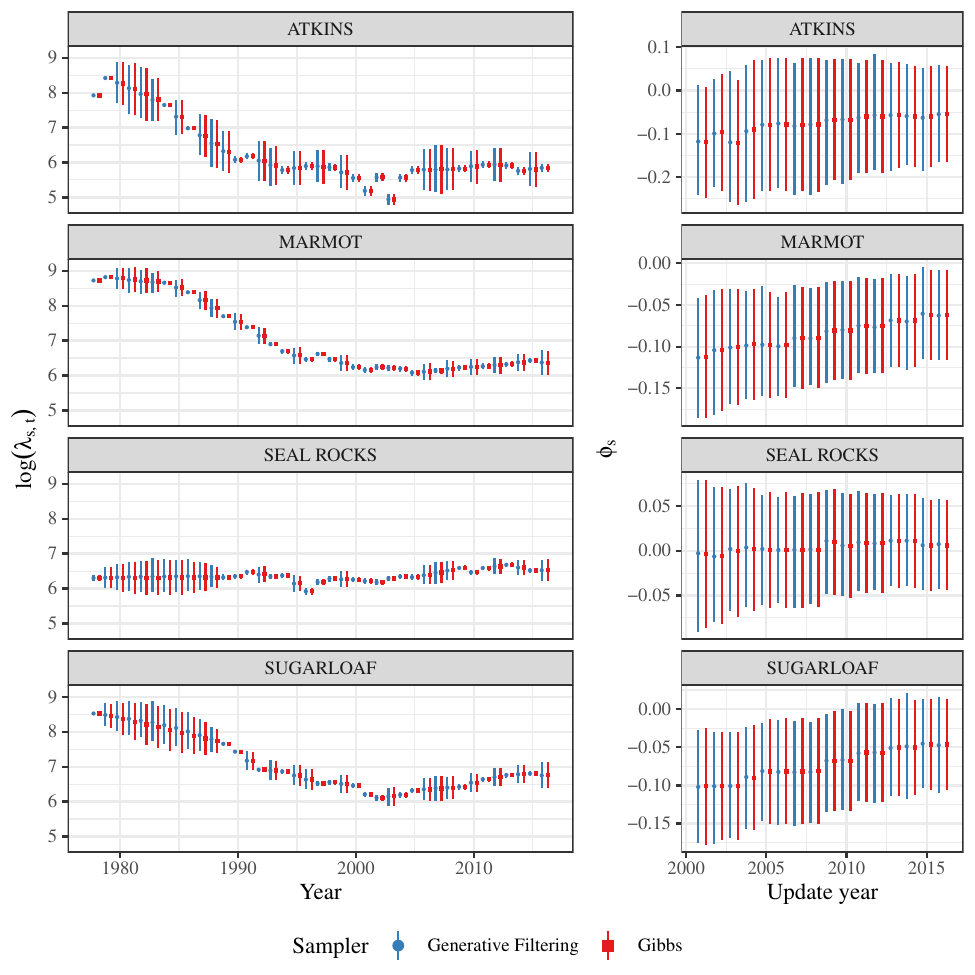}
\caption{\label{fig:pups-intervals-plot1}LEFT: Posterior means and credible intervals for \(\log(\lambda_{s,t})\) parameters after all updates. Credible intervals are wider in years where data were missing. In three of four sites, population intensity estimates declined, then began to rise after 2000. RIGHT: Posterior means and credible intervals for \(\phi_s\) following each streaming update. Means shift positive and credible intervals narrow with the arrival of new data. On both the left and right, means and credible intervals agree between Gibbs and Generative Filtering.}
\end{figure}

We assume that counts of sea lion pups are produced from latent population intensities for each site-year combination, which are related through a log-scale autoregressive process. We model these data using the following hierarchical model,
\begin{align*}
y_{s,t} &\sim \mbox{Pois}(\lambda_{s,t}) \\
\log(\lambda_{s,1}) &\sim \mbox{N}(\mu_1, \sigma_1^2) \\
\log(\lambda_{s,t}) &\sim \mbox{N}(\phi_s + \log(\lambda_{s,t-1}), \sigma^2_s) \\
\phi_s &\sim \mbox{N}(0, \sigma^2_\phi) \\
\sigma^2_s &\sim \mbox{Inverse-gamma}(\alpha, \beta),
\end{align*}
where \(s=1,2,3,4\) for each of our four studied sites, and \(t=1978, \dots, 2016\). The parameters \(\lambda_{s,t}\) represent the latent population intensities, and the parameters \(\phi_s\) and \(\sigma^2_s\) define the relationship between population intensity parameters over time. The parameters of interest are population intensities \(\lambda_{s,t}\) and population intensity trends \(\phi_s\). Negative values of \(\phi_s\) indicate the population intensities are decreasing at site \(s\) while positive values of \(\phi_s\) indicate the population intensities are increasing at site \(s\).

These data were analyzed by \citet{hooten2021making} who used PPRB to update a temporal model of two measured sites (Marmot and Sugarloaf) a single time from year 2013 to 2015. We extend this analysis in two ways. First, we include two additional sites (Seal Rocks and Atkins) from the same survey. Second, we perform a series of 16 streaming updates from the year 2000 to 2016. We assume a first-stage analysis has been performed for all data through year 2000 which will form the basis for the series of streaming updates performed. Each update incorporates the data from an additional year. We use the same hyperparameters as in \citet{hooten2021making}, specifically, \(\mu_1 = 8.7\), \(\sigma^2_1 = 1.69\), \(\sigma^2_\phi = 1\), \(\alpha = 1\), and \(\beta = 20\).

We perform the series of updates using four methods: a non-streaming Gibbs sampler for each year fitting the full model (Gibbs), sequential PPRB-within-Gibbs updates (PPRB-within-Gibbs), sequential SMCMC updates (SMCMC), and sequential Generative Filtering. The full sampler details can be found in Appendix \ref{sec:application-sampling-details}. Similarly to the simulation in Section \ref{sec:simulation-gaussian}, we perform transition kernel steps in SMCMC and Generative Filtering until a desired level of convergence to the posterior is reached. Unlike in Section \ref{sec:simulation-gaussian}, however, the true posterior is unknown in this model. Thus we use the Gibbs samples as a reference, and stop SMCMC and Generative filtering when the KS statistic comparing the current state of the ensemble to the Gibbs samples is below a desired threshold. This approach is not feasible in scenarios where streaming is required and we discuss practical options for choosing \(m_t\) in Section \ref{sec:discussion}.

\begin{figure}
\centering
\includegraphics{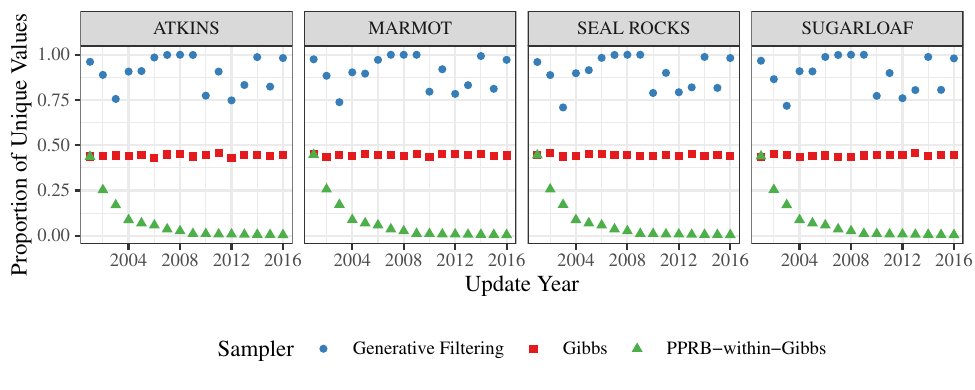}
\caption{\label{fig:pups-unique-values-plot}Number of unique values of \(\log(\lambda_{s,t})\), for \(t=2001\), present in samples as a proportion of the number of samples produced with each sampler. The unique values produced by PPRB-within-Gibbs decays with each successive update. The Gibbs sampler produces a consistent number of unique values, approximately 44\%, due to the tuning of its component updates. Generative Filtering produces consistently the highest proportion of unique values among its samples.}
\end{figure}

Figure \ref{fig:pups-intervals-plot1} (left) shows the posterior means and credible intervals for \(\log(\lambda_{s,t})\). In three of four sites, population intensity estimates declined, then began to rise after 2000. We are also able to recover estimates of \(\log(\lambda_{s,t})\) in years where there is missing data for the sites, though the credible intervals are wide for these parameters. Figure \ref{fig:pups-intervals-plot1} (right) shows the changing posterior mean and credible intervals for each parameter \(\phi_s\), with each Bayesian update. Corresponding to the trends observed in the estimates of \(\log(\lambda_{s,t})\), the posterior means and credible intervals of \(\phi_s\) become less negative with each new year of data. The inference for both \(\log(\lambda_{s,t})\) and \(\phi_s\) using Generative Filtering is nearly identical to non-streaming Gibbs, but with cumulatively 32\% less time required with only one core.

We compare the samplers' ability to obtain these estimates through three metrics: the number of distinct values present in repeatedly updated parameters (Figure \ref{fig:pups-unique-values-plot}), the number of transition kernel steps required by each of SMCMC and Generative Filtering (Figure \ref{fig:pups-runtime-steps-plot}, top), and the total time required to perform all updates for SMCMC, Generative Filtering, PPRB-within-Gibbs, and Gibbs with varying numbers of cores available (Figure \ref{fig:pups-runtime-steps-plot}, bottom). First, Generative filtering outperforms PPRB-within-Gibbs in terms of the proportion of unique values within its produced samples (Figure \ref{fig:pups-unique-values-plot}). Generative Filtering also has a higher proportion of unique values than Gibbs, due to the Gibbs sampler's component updates being tuned to 44\% acceptance rate. We also see that Generative Filtering outperforms SMCMC in both number of transition kernel steps required, and both SMCMC and Gibbs in runtime (Figure \ref{fig:pups-runtime-steps-plot}). In particular, two cases in the runtime comparison are especially interesting. For one available core, i.e., sequential execution, Generative Filtering is faster even than Gibbs. This indicates that even when no parallel execution is possible, Generative Filtering takes advantage of previously produced samples and is faster than refitting the model from scratch. For 1000 available cores, we simulate the situation where essentially unlimited parallel computing resources are available and each thread in either SMCMC or Generative Filtering is able to be run in parallel. In this case, SMCMC overtakes Generative Filtering in speed because its jumping kernel is able to be parallelized while PPRB-within-Gibbs is not. However, this is an extreme setting. We also note that both methods outperform Gibbs by wide margins, and that while PPRB is relatively fast for small numbers of available cores, it is undesirable due to its degradation. Together, these timing results suggest that Generative Filtering is an attractive method for streaming Bayesian updates where time is a consideration, but few or moderate computational resources are available.

\begin{figure}
\centering
\includegraphics{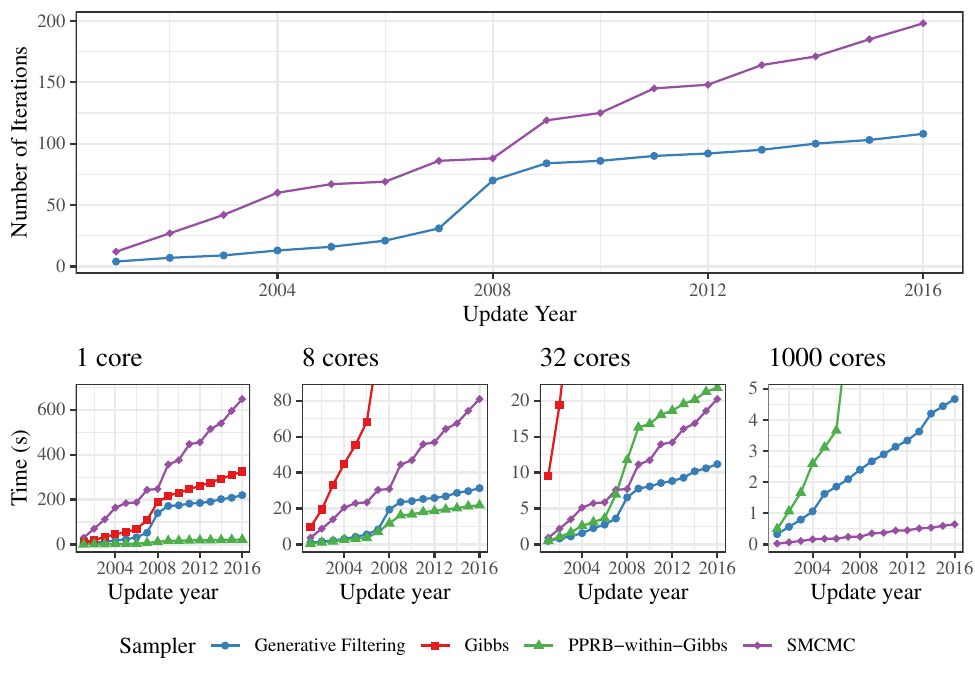}
\caption{\label{fig:pups-runtime-steps-plot}TOP: Total number of transition kernel steps required for each method to converge to the posterior distribution. Generative Filtering requires fewer transition kernel steps due to its initial PPRB-within-Gibbs step. BOTTOM: Cumulative time to reach 1000 effective sample size, for varying numbers of available parallel cores. Left to right: 1 core, 8 cores, 32 cores, 1000 cores. These represent, respectively, no available parallel resources, a typical laptop or workstation, a single node in a high performance computing environment, and unlimited parallel resources. Due to degradation, the PPRB-within-Gibbs samples are not necessarily from the posterior distribution.}
\end{figure}

\hypertarget{sec:discussion}{%
\section{Discussion}\label{sec:discussion}}

We define Generative Filtering as an efficient way to perform recursive Bayesian updates in a streaming data context when moderate parallel computing resources are available. We show that Generative Filtering resolves the problems of the two methods it extends, avoiding the degradation of filtering methods such as PPRB and PPRB-within-Gibbs and converging faster with fewer resources than SMCMC. We characterize the degradation of PPRB and PPRB-within-Gibbs samples after repeated application, and provide novel bounds on the error introduced by this degradation. We additionally provide sufficient conditions for reducing the data storage needs of Generative Filtering.

We conduct two simulation studies with streaming data models to demonstrate the effectiveness of Generative Filtering. We find that repeated application of PPRB-within-Gibbs in a streaming setting leads to a measurable accumulation of MCMC error and show that Generative Filtering avoids this error to reach the same level of convergence as SMCMC using less time with moderate computational resources.

We use Generative Filtering to analyze Steller sea lion pup counts discovering trends in count intensity at four different sites and observing changes in parameter estimates with the arrival of new data. We find that Generative Filtering required fewer transition kernel steps than SMCMC and less time with moderate computation resources (32 cores or fewer). We also find that Generative Filtering requires less time to produce the equivalent of 1000 independent samples than Gibbs when only one core is available.

Throughout this paper, when comparing the time to convergence for streaming samplers, we have used an objective standard to determine when a given algorithm converged to the posterior distribution. We compared samples to either the true posterior distribution (if known), or to reference samples produced by a non-streaming alternative sampler. This worked for our purposes of demonstrating faster convergence for Generative Filtering than SMCMC. However, this approach is not relevant in practice and determining the number of transition kernel iterations to perform remains an open question. We have compared to the correlation-based stopping of \citet{yang2013sequential} in an appropriate example, however, this approach is not appropriate for many cases, e.g., discrete parameters without a natural meaning or measure of correlation. Future work for samplers of this type will include a more general stopping criterion.

When the number of transition kernel steps, \(m_t\), is large, then it can be considered wasteful to discard all but the final value in each parallel chain. \citet{dau2022waste} propose Waste-Free SMC, in which a subset of the values from an SMC method are saved, from which parallel MCMC chains are run and all values in each chain are saved. Such a modification may be made to Generative Filtering to avoid wasted computation of the transition kernel.

\hypertarget{funding-acknowledgement}{%
\section*{Funding Acknowledgement}\label{funding-acknowledgement}}
\addcontentsline{toc}{section}{Funding Acknowledgement}

B. Betancourt acknowledges the support of NSF DMS-2310222. A. Kaplan acknowledges the support of NSF DMS-2330089 and NSF SES-2338428.

\bibliographystyle{Chicago}
\bibliography{bibliography.bib}

\begin{thebibliography}{}

\bibitem[\protect\citeauthoryear{Aleshin-Guendel and Sadinle}{Aleshin-Guendel
  and Sadinle}{2023}]{aleshinguendel2021multifile}
Aleshin-Guendel, S. and M.~Sadinle (2023).
\newblock Multifile partitioning for record linkage and duplicate detection.
\newblock {\em Journal of the American Statistical Association\/}~{\em
  118\/}(543), 1786--1795.

\bibitem[\protect\citeauthoryear{Altwaijry, Kalashnikov, and
  Mehrotra}{Altwaijry et~al.}{2017}]{altwaijry2017qda}
Altwaijry, H., D.~V. Kalashnikov, and S.~Mehrotra (2017).
\newblock Qda: A query-driven approach to entity resolution.
\newblock {\em IEEE Transactions on Knowledge and Data Engineering\/}~{\em
  29\/}(2), 402--417.

\bibitem[\protect\citeauthoryear{Berzuini, Best, Gilks, and Larizza}{Berzuini
  et~al.}{1997}]{berzuini1997dynamic}
Berzuini, C., N.~G. Best, W.~R. Gilks, and C.~Larizza (1997).
\newblock Dynamic conditional independence models and markov chain monte carlo
  methods.
\newblock {\em Journal of the American Statistical Association\/}~{\em
  92\/}(440), 1403--1412.

\bibitem[\protect\citeauthoryear{Blei, Kucukelbir, and McAuliffe}{Blei
  et~al.}{2017}]{blei2017variational}
Blei, D.~M., A.~Kucukelbir, and J.~D. McAuliffe (2017).
\newblock Variational inference: A review for statisticians.
\newblock {\em Journal of the American Statistical Association\/}~{\em
  112\/}(518), 859--877.

\bibitem[\protect\citeauthoryear{Chopin}{Chopin}{2002}]{chopin2002sequential}
Chopin, N. (2002).
\newblock {A sequential particle filter method for static models}.
\newblock {\em Biometrika\/}~{\em 89\/}(3), 539--552.

\bibitem[\protect\citeauthoryear{Christen, Gayler, and Hawking}{Christen
  et~al.}{2009}]{christen2009similarity}
Christen, P., R.~Gayler, and D.~Hawking (2009).
\newblock Similarity-aware indexing for real-time entity resolution.
\newblock CIKM '09, New York, NY, USA, pp.\  1565–1568. Association for
  Computing Machinery.

\bibitem[\protect\citeauthoryear{Dau and Chopin}{Dau and
  Chopin}{2021}]{dau2022waste}
Dau, H.-D. and N.~Chopin (2021).
\newblock {Waste-Free Sequential Monte Carlo}.
\newblock {\em Journal of the Royal Statistical Society Series B: Statistical
  Methodology\/}~{\em 84\/}(1), 114--148.

\bibitem[\protect\citeauthoryear{Daum}{Daum}{1986}]{daum1986fisher}
Daum, F.~E. (1986).
\newblock The fisher-darmois-koopman-pitman theorem for random processes.
\newblock In {\em 1986 25th IEEE Conference on Decision and Control}, pp.\
  1043--1044.

\bibitem[\protect\citeauthoryear{Del~Moral, Doucet, and Jasra}{Del~Moral
  et~al.}{2006}]{delmoral2006sequential}
Del~Moral, P., A.~Doucet, and A.~Jasra (2006).
\newblock {Sequential Monte Carlo Samplers}.
\newblock {\em Journal of the Royal Statistical Society Series B: Statistical
  Methodology\/}~{\em 68\/}(3), 411--436.

\bibitem[\protect\citeauthoryear{Dey, Mookerjee, and Liu}{Dey
  et~al.}{2011}]{dey2011online}
Dey, D., V.~Mookerjee, and D.~Liu (2011).
\newblock Efficient techniques for online record linkage.
\newblock {\em IEEE Transactions on Knowledge and Data Engineering\/}~{\em
  23\/}(3), 373--387.

\bibitem[\protect\citeauthoryear{Duane, Kennedy, Pendleton, and Roweth}{Duane
  et~al.}{1987}]{duane1987hybrid}
Duane, S., A.~Kennedy, B.~J. Pendleton, and D.~Roweth (1987).
\newblock Hybrid monte carlo.
\newblock {\em Physics Letters B\/}~{\em 195\/}(2), 216--222.

\bibitem[\protect\citeauthoryear{Fritz, Ferrero, and Berg}{Fritz
  et~al.}{1995}]{fritz1995threatened}
Fritz, L.~W., R.~C. Ferrero, and R.~J. Berg (1995).
\newblock {The Threatened Status of Steller Sea Lions, Eumetopias jubatus,
  under the Endangered Species Act: Effects on Alaska Groundfish Fisheries
  Management}.
\newblock {\em Marine Fisheries Review\/}~{\em 57\/}(2), 14--27.

\bibitem[\protect\citeauthoryear{Gelfand and Smith}{Gelfand and
  Smith}{1990}]{gelfand1990sampling}
Gelfand, A.~E. and A.~F.~M. Smith (1990).
\newblock Sampling-based approaches to calculating marginal densities.
\newblock {\em Journal of the American Statistical Association\/}~{\em
  85\/}(410), 398--409.

\bibitem[\protect\citeauthoryear{Gelman, Gilks, and Roberts}{Gelman
  et~al.}{1997}]{gelman1997weak}
Gelman, A., W.~R. Gilks, and G.~O. Roberts (1997).
\newblock {Weak convergence and optimal scaling of random walk Metropolis
  algorithms}.
\newblock {\em The Annals of Applied Probability\/}~{\em 7\/}(1), 110 -- 120.

\bibitem[\protect\citeauthoryear{Gilks and Berzuini}{Gilks and
  Berzuini}{2001}]{gilks2001following}
Gilks, W.~R. and C.~Berzuini (2001).
\newblock {Following a moving target—Monte Carlo inference for dynamic
  Bayesian models}.
\newblock {\em Journal of the Royal Statistical Society: Series B (Statistical
  Methodology)\/}~{\em 63\/}(1), 127--146.

\bibitem[\protect\citeauthoryear{Gordon, Salmond, and Smith}{Gordon
  et~al.}{1993}]{gordon1993novel}
Gordon, N., D.~Salmond, and A.~Smith (1993).
\newblock Novel approach to nonlinear/non-gaussian bayesian state estimation.
\newblock {\em IEE Proceedings F (Radar and Signal Processing)\/}~{\em 140},
  107--113.

\bibitem[\protect\citeauthoryear{Haario, Saksman, and Tamminen}{Haario
  et~al.}{2001}]{haario2001adaptive}
Haario, H., E.~Saksman, and J.~Tamminen (2001).
\newblock {An adaptive Metropolis algorithm}.
\newblock {\em Bernoulli\/}~{\em 7\/}(2), 223 -- 242.

\bibitem[\protect\citeauthoryear{Hendry and Richard}{Hendry and
  Richard}{1992}]{hendry1992likelihood}
Hendry, D.~F. and J.-F. Richard (1992).
\newblock {\em {Likelihood Evaluation for Dynamic Latent Variables Models}},
  pp.\  3--17.
\newblock Dordrecht: Springer Netherlands.

\bibitem[\protect\citeauthoryear{Hoffman and Gelman}{Hoffman and
  Gelman}{2014}]{hoffman2014no}
Hoffman, M.~D. and A.~Gelman (2014).
\newblock {The No-U-Turn sampler: adaptively setting path lengths in
  Hamiltonian Monte Carlo}.
\newblock {\em Journal of Machine Learning Research\/}~{\em 15\/}(1),
  1593--1623.

\bibitem[\protect\citeauthoryear{Hooten, Johnson, and Brost}{Hooten
  et~al.}{2021}]{hooten2021making}
Hooten, M.~B., D.~S. Johnson, and B.~M. Brost (2021).
\newblock Making recursive bayesian inference accessible.
\newblock {\em The American Statistician\/}~{\em 75\/}(2), 185--194.

\bibitem[\protect\citeauthoryear{Ioannou, Nejdl, Nieder\'{e}e, and
  Velegrakis}{Ioannou et~al.}{2010}]{ioannou2010onthefly}
Ioannou, E., W.~Nejdl, C.~Nieder\'{e}e, and Y.~Velegrakis (2010).
\newblock {On-the-fly entity-aware query processing in the presence of
  linkage}.
\newblock {\em Proceedings of the VLDB Endowment\/}~{\em 3\/}(1–2), 429--438.

\bibitem[\protect\citeauthoryear{Johnson}{Johnson}{2017}]{pkgagTrend}
Johnson, D.~S. (2017).
\newblock {\em agTrend: Estimate Linear Trends for Aggregated Abundance Data}.
\newblock R package version 0.17.7.

\bibitem[\protect\citeauthoryear{Kantas, Doucet, Singh, and Maciejowski}{Kantas
  et~al.}{2009}]{kantas2009overview}
Kantas, N., A.~Doucet, S.~Singh, and J.~Maciejowski (2009).
\newblock An overview of sequential monte carlo methods for parameter
  estimation in general state-space models.
\newblock {\em IFAC Proceedings Volumes\/}~{\em 42\/}(10), 774--785.
\newblock 15th IFAC Symposium on System Identification.

\bibitem[\protect\citeauthoryear{Karapiperis, Gkoulalas{-}Divanis, and
  Verykios}{Karapiperis et~al.}{2018}]{karapiperis2018summarization}
Karapiperis, D., A.~Gkoulalas{-}Divanis, and V.~S. Verykios (2018).
\newblock {Summarization Algorithms for Record Linkage}.
\newblock In M.~H. B{\"{o}}hlen, R.~Pichler, N.~May, E.~Rahm, S.~Wu, and
  K.~Hose (Eds.), {\em {Proceedings of the 21st International Conference on
  Extending Database Technology, {EDBT} 2018, Vienna, Austria, March 26-29,
  2018}}, pp.\  73--84. OpenProceedings.org.

\bibitem[\protect\citeauthoryear{Kitagawa}{Kitagawa}{1996}]{kitagawa1996monte}
Kitagawa, G. (1996).
\newblock {Monte Carlo Filter and Smoother for Non-Gaussian Nonlinear State
  Space Models}.
\newblock {\em Journal of Computational and Graphical Statistics\/}~{\em
  5\/}(1), 1--25.

\bibitem[\protect\citeauthoryear{Liu and Chen}{Liu and
  Chen}{1995}]{liu1995blind}
Liu, J.~S. and R.~Chen (1995).
\newblock {Blind Deconvolution via Sequential Imputations}.
\newblock {\em Journal of the American Statistical Association\/}~{\em
  90\/}(430), 567--576.

\bibitem[\protect\citeauthoryear{Maceachern, Clyde, and Liu}{Maceachern
  et~al.}{1999}]{maceachern1999sequential}
Maceachern, S.~N., M.~Clyde, and J.~S. Liu (1999).
\newblock {Sequential importance sampling for nonparametric Bayes models: The
  next generation}.
\newblock {\em Canadian Journal of Statistics\/}~{\em 27\/}(2), 251--267.

\bibitem[\protect\citeauthoryear{Neal}{Neal}{2001}]{neal2001annealed}
Neal, R.~M. (2001).
\newblock Annealed importance sampling.
\newblock {\em Statistics and Computing\/}~{\em 11\/}(2), 125--139.

\bibitem[\protect\citeauthoryear{Rubin}{Rubin}{1988}]{rubin1988using}
Rubin, D.~B. (1988).
\newblock {\em Using the {SIR} algorithm to simulate posterior distributions},
  pp.\  395--402.
\newblock Oxford University Press.

\bibitem[\protect\citeauthoryear{Sadinle}{Sadinle}{2017}]{sadinle2017bayesian}
Sadinle, M. (2017).
\newblock {Bayesian Estimation of Bipartite Matchings for Record Linkage}.
\newblock {\em Journal of the American Statistical Association\/}~{\em
  112\/}(518), 600--612.

\bibitem[\protect\citeauthoryear{S{\"a}rkk{\"a}}{S{\"a}rkk{\"a}}{2013}]{sarkka2013bayesian}
S{\"a}rkk{\"a}, S. (2013).
\newblock {\em Bayesian Filtering and Smoothing}.
\newblock Institute of Mathematical Statistics Textbooks. Cambridge University
  Press.

\bibitem[\protect\citeauthoryear{Smith and Gelfand}{Smith and
  Gelfand}{1992}]{smith1992bayesian}
Smith, A. F.~M. and A.~E. Gelfand (1992).
\newblock {Bayesian Statistics without Tears: A Sampling–Resampling
  Perspective}.
\newblock {\em The American Statistician\/}~{\em 46\/}(2), 84--88.

\bibitem[\protect\citeauthoryear{Steorts, Hall, and Fienberg}{Steorts
  et~al.}{2016}]{steorts2016bayesian}
Steorts, R.~C., R.~Hall, and S.~E. Fienberg (2016).
\newblock {A Bayesian Approach to Graphical Record Linkage and Deduplication}.
\newblock {\em Journal of the American Statistical Association\/}~{\em
  111\/}(516), 1660--1672.

\bibitem[\protect\citeauthoryear{Tancredi and Liseo}{Tancredi and
  Liseo}{2011}]{tancredi2011hierarchical}
Tancredi, A. and B.~Liseo (2011).
\newblock {A hierarchical Bayesian approach to record linkage and population
  size problems}.
\newblock {\em The Annals of Applied Statistics\/}~{\em 5\/}(2B), 1553 -- 1585.

\bibitem[\protect\citeauthoryear{Tavaré, Balding, Griffiths, and
  Donnelly}{Tavaré et~al.}{1997}]{tavare1997inferring}
Tavaré, S., D.~J. Balding, R.~C. Griffiths, and P.~Donnelly (1997).
\newblock {Inferring Coalescence Times From DNA Sequence Data}.
\newblock {\em Genetics\/}~{\em 145\/}(2), 505--518.

\bibitem[\protect\citeauthoryear{Taylor, Kaplan, and Betancourt}{Taylor
  et~al.}{2022}]{pkgbstrl}
Taylor, I., A.~Kaplan, and B.~Betancourt (2022).
\newblock {\em bstrl: Bayesian Streaming Record Linkage}.
\newblock R package version 1.0.2.

\bibitem[\protect\citeauthoryear{Taylor, Kaplan, and Betancourt}{Taylor
  et~al.}{2024}]{taylor2023fast}
Taylor, I., A.~Kaplan, and B.~Betancourt (2024).
\newblock {Fast Bayesian Record Linkage for Streaming Data Contexts}.
\newblock {\em Journal of Computational and Graphical Statistics\/}~{\em
  33\/}(3), 833--844.

\bibitem[\protect\citeauthoryear{Tran, Vatsalan, and Christen}{Tran
  et~al.}{2013}]{tran2013geco}
Tran, K.-N., D.~Vatsalan, and P.~Christen (2013).
\newblock {GeCo: an online personal data generator and corruptor}.
\newblock In {\em Proceedings of the 22nd ACM International Conference on
  Information \& Knowledge Management}, CIKM '13, New York, NY, USA, pp.\
  2473--2476. Association for Computing Machinery.

\bibitem[\protect\citeauthoryear{Yang and Dunson}{Yang and
  Dunson}{2013}]{yang2013sequential}
Yang, Y. and D.~B. Dunson (2013).
\newblock {Sequential Markov Chain Monte Carlo}.

\bibitem[\protect\citeauthoryear{Zanella}{Zanella}{2020}]{zanella2020informed}
Zanella, G. (2020).
\newblock {Informed Proposals for Local MCMC in Discrete Spaces}.
\newblock {\em Journal of the American Statistical Association\/}~{\em
  115\/}(530), 852--865.

\end{thebibliography}

\vfill

\newpage
\spacingset{1.75} % DON'T change the spacing!

\hypertarget{appendix-appendix}{%
\appendix}

\hypertarget{algorithm-specifications}{%
\section{Algorithm Specifications}\label{algorithm-specifications}}

Algorithm \ref{alg:pprb-within-gibbs-2step} provides a formal algorithm specification for the two step PPRB-within-Gibbs algorithm (Definition \ref{def:pprb-within-gibbs-2step}).

\begin{algorithm}
\caption{Two step PPRB-within-Gibbs}\label{alg:pprb-within-gibbs-2step}
\algrenewcommand\algorithmicrequire{\textbf{Input:}}
\algrenewcommand\algorithmicensure{\textbf{Output:}}
\begin{algorithmic}[1]
\Require Posterior sample $\{\boldsymbol \theta_s\}_{s=1}^S \sim p(\boldsymbol \theta | \boldsymbol y_1)$
\Require Data $\bm{y}_{2} \sim p(\bm{y}_2 | \bm{\theta}, \bm{\phi}, \bm{y}_1)$
\Ensure Updated posterior sample $\{(\boldsymbol \theta'_s, \boldsymbol \phi_s)\}_{s=1}^S \sim p(\boldsymbol \theta, \boldsymbol \phi | \boldsymbol y_1, \boldsymbol y_2)$
    
    \State Draw $s^\ast \sim \text{Uniform}(\{1,\dots,S\})$
    \State Set $\boldsymbol \theta'_0 = \boldsymbol \theta_{s^\ast}$
    \State Draw $\boldsymbol \phi_0 \sim p(\boldsymbol \phi| \boldsymbol \theta'_0)$
    \For{$s=1,\dots,S$}
        \State Draw $s^\ast \sim \text{Uniform}(\{1,\dots,S\})$
        \State Set $\boldsymbol \theta^\ast = \boldsymbol \theta_{s^\ast}$
        \State Set $\alpha = \min\left(\frac{p(\boldsymbol y_2 | \boldsymbol y_1, \boldsymbol \theta^\ast, \boldsymbol \phi)}{p(\boldsymbol y_2 | \boldsymbol y_1, \boldsymbol \theta'_{s-1}, \boldsymbol \phi)}\frac{p(\boldsymbol \phi | \boldsymbol \theta^\ast)}{p(\boldsymbol \phi | \boldsymbol \theta'_{s-1})}, 1\right)$
        \State Draw $p \sim \text{Uniform}([0,1])$.
        \If{$p < \alpha$}
            \State Set $\boldsymbol \theta'_s = \boldsymbol \theta^\ast$
        \Else
            \State Set $\boldsymbol \theta'_s = \boldsymbol \theta'_{s-1}$
        \EndIf
        \State Draw $\boldsymbol \phi_s \sim p(\boldsymbol \phi|\boldsymbol \theta'_s, \boldsymbol y_1, \boldsymbol y_2)$ \label{algstep:pprb-within-gibbs-phi-update}
    \EndFor

\end{algorithmic}
\end{algorithm}

Algorithm \ref{alg:generative-filtering} provides a formal algorithm specification for Generative Filtering. As noted, we use PPRB-within-Gibbs for the filtering method in step (\ref{algstep:gf-filtering}) throughout this paper. However, any filtering algorithm may be substituted. SMCMC can be obtained by replacing step (\ref{algstep:gf-filtering}) by the parallel use of a jumping kernel to initialize \(\bm{\theta}_{t}\). The loop over the sample in step (\ref{algstep:gf-parallel-chains}) may be performed in parallel as each chain is independent of the others.

\begin{algorithm}
\caption{Generative Filtering}\label{alg:generative-filtering}
\algrenewcommand\algorithmicrequire{\textbf{Input:}}
\algrenewcommand\algorithmicensure{\textbf{Output:}}
\begin{algorithmic}[1]
\Require Posterior sample $\{\bm{\theta}_{1:(t-1),s}\}_{s=1}^S \sim p(\bm{\theta}_{1:(t-1)} | \bm{y}_{1:(t-1)})$
\Require Data $\bm{y}_{t} \sim p(\bm{y}_t | \bm{\theta}_{1:t})$
\Require Filtering method \Call{Filter}{} approximately targeting $p(\bm{\theta}_{1:t} | \bm{y}_{1:t})$, e.g., Algorithm \ref{alg:pprb-within-gibbs-2step}
\Require Markov chain transition kernel $T_t(\cdot, \cdot)$, targeting $p(\bm{\theta}_{1:t} | \bm{y}_{1:t})$
\Require $m_t \geq 1$
\Ensure Updated posterior sample $\{\bm{\theta}_{1:t,s}\}_{s=1}^S \sim p(\bm{\theta}_{1:t} | \bm{y}_{1:t})$

    \State Set $\{\bm{\theta}^{\ast}_{1:t,s}\}_{s=1}^S$ to be the result of $\Call{Filter}{\{\bm{\theta}_{1:(t-1),s}\}_{s=1}^S, \bm{y}_{t}}$ \label{algstep:gf-filtering}
    \For{$s=1,\dots,S$} \label{algstep:gf-parallel-chains}
        \State Set $\bm{\theta}_{1:t,s} = \bm{\theta}^{\ast}_{1:t,s}$
        \For{$i=1,\dots,m_t$}
            \State Draw $\bm{\theta}_{1:t,s}$ from the density $T_t(\bm{\theta}_{1:t,s}, \cdot)$ \label{algstep:generative-filtering-transition-kernel}
        \EndFor
    \EndFor

\end{algorithmic}
\end{algorithm}

\hypertarget{sec:supplement-theorems}{%
\section{Theorems and Proofs}\label{sec:supplement-theorems}}

\hypertarget{pprb-within-gibbs}{%
\subsection{PPRB-within-Gibbs}\label{pprb-within-gibbs}}

\begin{theorem}
\protect\hypertarget{thm:pprb-within-gibbs-sampling}{}\label{thm:pprb-within-gibbs-sampling}The PPRB-within-Gibbs sampler (Definition \ref{def:pprb-within-gibbs-2step}) produces an ergodic Markov chain with the model's posterior distribution as its target distribution if the posterior distribution satisfies the following positivity condition,
\[p(\boldsymbol \theta | \boldsymbol y_1, \boldsymbol y_2) > 0,\ p(\boldsymbol \phi | \boldsymbol y_1, \boldsymbol y_2) > 0 \implies p(\boldsymbol \theta, \boldsymbol \phi | \boldsymbol y_1, \boldsymbol y_2) > 0.\]
\end{theorem}

\begin{proof}
First, we show that the acceptance ratio, \(\alpha\), in step 1 is appropriate for the target full conditional distribution. Since the proposals, \(\bm{\theta}^\ast\) are produced from the distribution \(p(\bm{\theta}^\ast|\bm{y}_1)\) and the target distribution is the full conditional, \(p(\bm{\theta}|\bm{\phi},\bm{y}_1,\bm{y}_2)\), the MH acceptance ratio would be
\begin{align*}
\alpha &= \frac{p(\bm{\theta}^\ast|\bm{\phi},\bm{y}_1,\bm{y}_2)}{p(\bm{\theta}|\bm{\phi},\bm{y}_1,\bm{y}_2)}\cdot\frac{p(\bm{\theta}|\bm{y}_1)}{p(\bm{\theta}^\ast|\bm{y}_1)} \\
&= \frac{p(\bm{\theta}^\ast,\bm{\phi}|\bm{y}_1,\bm{y}_2)}{p(\bm{\theta},\bm{\phi}|\bm{y}_1,\bm{y}_2)}\cdot\frac{p(\bm{\theta}|\bm{y}_1)}{p(\bm{\theta}^\ast|\bm{y}_1)} \\
&= \frac{p(\bm{y}_2|\bm{\phi},\bm{\theta}^\ast,\bm{y}_1)p(\bm{y}_1|\bm{\theta}^\ast)p(\bm{\phi}|\bm{\theta}^\ast)p(\bm{\theta}^\ast)}{p(\bm{y}_2|\bm{\phi},\bm{\theta},\bm{y}_1)p(\bm{y}_1|\bm{\theta})p(\bm{\phi}|\bm{\theta})p(\bm{\theta})}\cdot\frac{p(\bm{\theta}|\bm{y}_1)}{p(\bm{\theta}^\ast|\bm{y}_1)} \\
&= \frac{p(\bm{y}_2|\bm{\phi},\bm{\theta}^\ast,\bm{y}_1)p(\bm{\phi}|\bm{\theta}^\ast)\cdot p(\bm{\theta}^\ast|\bm{y}_1)}{p(\bm{y}_2|\bm{\phi},\bm{\theta},\bm{y}_1)p(\bm{\phi}|\bm{\theta}) \cdot p(\bm{\theta}|\bm{y}_1)}\cdot\frac{p(\bm{\theta}|\bm{y}_1)}{p(\bm{\theta}^\ast|\bm{y}_1)} \\
&= \frac{p(\boldsymbol y_2 | \boldsymbol y_1, \boldsymbol \theta^\ast, \boldsymbol \phi)}{p(\boldsymbol y_2 | \boldsymbol y_1, \boldsymbol \theta, \boldsymbol \phi)}\frac{p(\boldsymbol \phi | \boldsymbol \theta^\ast)}{p(\boldsymbol \phi | \boldsymbol \theta)}.
\end{align*}
Second, the distribution \(p(\bm{\theta}|\bm{y}_1)\) works as independent MH proposals for the distribution \(p(\bm{\theta}|\bm{y}_1,\bm{y}_2,\bm{\phi})\) because \(p(\bm{\theta}|\bm{y}_1,\bm{y}_2,\bm{\phi}) \propto p(\bm{\theta}|\bm{y}_1)p(\bm{y}_2|\bm{\theta},\bm{\phi},\bm{y}_1)p(\bm{\phi}|\bm{\theta})\), so \(p(\bm{\theta}|\bm{y}_1) = 0 \implies p(\bm{\theta}|\bm{y}_1,\bm{y}_2,\bm{\phi}) = 0\).

Finally, the positivity condition implies that a Gibbs sampler is irreducible \citep{robert2005monte}, so the PPRB-within-Gibbs algorithm produces an ergodic Markov chain.
\end{proof}

\begin{theorem}
\protect\hypertarget{thm:pprb-error-bounds-restated}{}\label{thm:pprb-error-bounds-restated}\textbf{(Theorem \ref{thm:pprb-error-bounds} restated)} Let \(\pi_t\), \(A_t\), and \(F_S^{(t)}\) be defined as in Eq. (\ref{eqn:inequality-terms-start})-(\ref{eqn:inequality-terms-end}) and let \(\|\cdot\|\) be a norm on probability measures. Then,
\begin{align*}
\left\|A_t - \pi_t\right\| &\leq \underbrace{\left\|A_{t-1} - \pi_{t-1}\right\|}_{(1)} + \underbrace{\left\|F_S^{(t-1)} - A_{t-1}\right\|}_{(2)} + \underbrace{\left\|\pi_t - \pi_{t-1}\right\|}_{(3)} + \underbrace{\left\|A_t - F_S^{(t-1)}\right\|}_{(4)} \\
\left\|A_t - \pi_t\right\| &\geq \bigg\lvert \underbrace{\left\|A_{t-1} - \pi_{t-1}\right\|}_{(1)} - \underbrace{\left\|F_S^{(t-1)} - A_{t-1}\right\|}_{(2)} \bigg\rvert - \underbrace{\left\|\pi_t - \pi_{t-1}\right\|}_{(3)} - \underbrace{\left\|A_t - F_S^{(t-1)}\right\|}_{(4)} 
\end{align*}
\end{theorem}

\begin{proof}
We derive the upper bound by the triangle inequality as
\begin{align}
\left\|A_t - \pi_t\right\| &\leq \left\|A_t - F_S^{(t-1)}\right\| + \left\|F_S^{(t-1)} - A_{t-1}\right\| + \left\|A_{t-1} - \pi_{t-1}\right\| + \left\|\pi_t - \pi_{t-1}\right\| \\
&= \left\|A_{t-1} - \pi_{t-1}\right\| + \left\|F_S^{(t-1)} - A_{t-1}\right\| + \left\|\pi_t - \pi_{t-1}\right\| + \left\|A_t - F_S^{(t-1)}\right\|. \label{eqn:pprb-error-triangle-2}
\end{align}

We use reverse triangle inequalities to derive the lower bound:
\begin{align}
\left\|A_t - \pi_t\right\| &\geq \bigg\lvert \left\|A_t - A_{t-1}\right\| - \left\|A_{t-1} - \pi_t\right\| \bigg\rvert \label{eqn:pprb-error-lower-work1} \\
&= \max\begin{Bmatrix*}[l]
\left\|A_t - A_{t-1}\right\| - \left\|A_{t-1} - \pi_t\right\|, \\
\left\|A_{t-1} - \pi_t\right\| - \left\|A_t - A_{t-1}\right\|
\end{Bmatrix*} \label{eqn:pprb-error-lower-work2} \\
&\geq \max\begin{Bmatrix*}[l]
\left\|A_t - A_{t-1}\right\| - \left(\left\|A_{t-1} - \pi_{t-1}\right\| + \left\|\pi_{t-1} - \pi_t\right\|\right), \\
\left\|A_{t-1} - \pi_t\right\| - \left(\left\|A_t - F_S^{(t-1)}\right\| + \left\|F_S^{(t-1)} - A_{t-1}\right\|\right) 
\end{Bmatrix*} \label{eqn:pprb-error-lower-work3} \\
&\geq \max\begin{Bmatrix*}[l]
\bigg\lvert \left\|A_t - F_S^{(t-1)}\right\| - \left\|F_S^{(t-1)} - A_{t-1}\right\| \bigg\rvert - \left\|A_{t-1} - \pi_{t-1}\right\| - \left\|\pi_{t-1} - \pi_t\right\|, \\
\bigg\lvert \left\|A_{t-1} - \pi_{t-1}\right\| - \left\|\pi_{t-1} - \pi_t\right\| \bigg\rvert - \left\|A_t - F_S^{(t-1)}\right\| - \left\|F_S^{(t-1)} - A_{t-1}\right\| 
\end{Bmatrix*} \label{eqn:pprb-error-lower-work4} \\
&= \max\begin{Bmatrix*}[l]
\left\|A_t - F_S^{(t-1)}\right\| - \left\|F_S^{(t-1)} - A_{t-1}\right\| - \left\|A_{t-1} - \pi_{t-1}\right\| - \left\|\pi_{t-1} - \pi_t\right\|, \\
\left\|F_S^{(t-1)} - A_{t-1}\right\| - \left\|A_t - F_S^{(t-1)}\right\| - \left\|A_{t-1} - \pi_{t-1}\right\| - \left\|\pi_{t-1} - \pi_t\right\|, \\
\left\|A_{t-1} - \pi_{t-1}\right\| - \left\|\pi_{t-1} - \pi_t\right\| - \left\|A_t - F_S^{(t-1)}\right\| - \left\|F_S^{(t-1)} - A_{t-1}\right\|, \\
\left\|\pi_{t-1} - \pi_t\right\| - \left\|A_{t-1} - \pi_{t-1}\right\| - \left\|A_t - F_S^{(t-1)}\right\| - \left\|F_S^{(t-1)} - A_{t-1}\right\| 
\end{Bmatrix*} \label{eqn:pprb-error-lower-work5} \\
&\geq \max\begin{Bmatrix*}[l]
\left\|F_S^{(t-1)} - A_{t-1}\right\| - \left\|A_{t-1} - \pi_{t-1}\right\| - \left\|\pi_{t-1} - \pi_t\right\| - \left\|A_t - F_S^{(t-1)}\right\|, \\
\left\|A_{t-1} - \pi_{t-1}\right\| - \left\|F_S^{(t-1)} - A_{t-1}\right\| - \left\|\pi_{t-1} - \pi_t\right\| - \left\|A_t - F_S^{(t-1)}\right\|
\end{Bmatrix*} \label{eqn:pprb-error-lower-work6} \\
&= \bigg\lvert \left\|A_{t-1} - \pi_{t-1}\right\| - \left\|F_S^{(t-1)} - A_{t-1}\right\| \bigg\rvert - \left\|\pi_{t-1} - \pi_t\right\| - \left\|A_t - F_S^{(t-1)}\right\| \label{eqn:pprb-error-lower}
\end{align}
Line (\ref{eqn:pprb-error-lower-work2}) is true by expanding \(|x| = \max\{x, -x\}\). Line (\ref{eqn:pprb-error-lower-work3}) follows by using a triangle inequality on the negative terms. Line (\ref{eqn:pprb-error-lower-work4}) follows by using a reverse triangle inequality on the positive terms. Line (\ref{eqn:pprb-error-lower-work5}) comes from expanding the abolute value as before. Line (\ref{eqn:pprb-error-lower-work6}) is true because \(\max A \geq \max B\) if \(B \subset A\). Finally, line (\ref{eqn:pprb-error-lower}) recombines the maximum into an absolute value.
\end{proof}

\hypertarget{generative-filtering}{%
\subsection{Generative Filtering}\label{generative-filtering}}

\begin{lemma}
\protect\hypertarget{lem:lemma38a}{}\label{lem:lemma38a}Let \(P^S_t(\bm{\theta}_{1:(t-1)},\cdot)\) represent the kernel resulting from \(S\) applications of PPRB-within-Gibbs at time \(t\), which is a probability density for \(\bm{\theta}_{1:t} := (\bm{\theta}_{1:(t-1)}, \bm{\theta}_t)\). Then for any probability density \(p(\cdot)\) for \(\bm{\theta}_{1:(t-1)}\), the following holds:
\[||\pi_t - P^S_t \circ p||_1 \leq \underset{\bm{\theta}_{1:(t-1)}}{\sup}||\pi_t - P^S_t(\bm{\theta}_{1:(t-1)},\cdot)||_1.\]
\end{lemma}

\begin{proof}
\begin{align}
||\pi_t - P^S_t \circ p||_1 &= \int\left|\pi_t(\bm{\theta}_{1:t}) - \int p(\bm{\theta}_{1:(t-1)}^\prime)P^S_t(\bm{\theta}_{1:(t-1)}^\prime,\bm{\theta}_{1:t}) \mathrm{d}\bm{\theta}_{1:(t-1)}^\prime\right|\mathrm{d}\bm{\theta}_{1:t} \\
&= \int\left|\int \left( p(\bm{\theta}_{1:(t-1)}^\prime)\pi_t(\bm{\theta}_{1:t}) - p(\bm{\theta}_{1:(t-1)}^\prime)P^S_t(\bm{\theta}_{1:(t-1)}^\prime,\bm{\theta}_{1:t})\right) \mathrm{d}\bm{\theta}_{1:(t-1)}^\prime\right|\mathrm{d}\bm{\theta}_{1:t} \\
&\leq \int\int \left| p(\bm{\theta}_{1:(t-1)}^\prime)\pi_t(\bm{\theta}_{1:t}) - p(\bm{\theta}_{1:(t-1)}^\prime)P^S_t(\bm{\theta}_{1:(t-1)}^\prime,\bm{\theta}_{1:t}) \right| \mathrm{d}\bm{\theta}_{1:(t-1)}^\prime\mathrm{d}\bm{\theta}_{1:t} \\
&= \int\int p(\bm{\theta}_{1:(t-1)}^\prime) \left| \pi_t(\bm{\theta}_{1:t}) - P^S_t(\bm{\theta}_{1:(t-1)}^\prime,\bm{\theta}_{1:t}) \right| \mathrm{d}\bm{\theta}_{1:(t-1)}^\prime\mathrm{d}\bm{\theta}_{1:t} \\
&\leq \int\int p(\bm{\theta}_{1:(t-1)}^\prime) \underset{\bm{\theta}_{1:(t-1)}^\prime}{\sup}\left| \pi_t(\bm{\theta}_{1:t}) - P^S_t(\bm{\theta}_{1:(t-1)}^\prime,\bm{\theta}_{1:t}) \right| \mathrm{d}\bm{\theta}_{1:(t-1)}^\prime\mathrm{d}\bm{\theta}_{1:t} \\
&= \int \underset{\bm{\theta}_{1:(t-1)}^\prime}{\sup}\left| \pi_t(\bm{\theta}_{1:t}) - P^S_t(\bm{\theta}_{1:(t-1)}^\prime,\bm{\theta}_{1:t}) \right| \mathrm{d}\bm{\theta}_{1:t}  \label{eqn:lemma38a-pre-switch-sup-int}\\
&= \underset{\bm{\theta}_{1:(t-1)}^\prime}{\sup}||\pi_t - P^S_t(\bm{\theta}_{1:(t-1)}^\prime,\cdot)||_1 \label{eqn:lemma38a-switch-sup-int}
\end{align}
\end{proof}

\begin{theorem}
\protect\hypertarget{thm:theorem39-restated}{}\label{thm:theorem39-restated}\textbf{(Theorem \ref{thm:theorem39} restated)} Let \(P^S_t(\bm{\theta}_{1:(t-1)},\cdot)\) represent the kernel resulting from \(S\) applications of a filtering method at time \(t\), which is a probability density for \(\bm{\theta}_{1:t} := (\bm{\theta}_{1:(t-1)}, \bm{\theta}_t)\). Let \(\pi_t = p(\bm{\theta}_{1:t}|\bm{y}_{1:t})\) be the target posterior at time \(t\). Assuming the following conditions:

\begin{enumerate}
\def\labelenumi{\arabic{enumi}.}
\tightlist
\item
  (Universal ergodicity) There exist \(\rho_t \in (0,1)\), such that for all \(t>0\) and \(x \in {\cal X}\), \[||T_t(x, \cdot) - \pi_t||_1 \leq 2\rho_t.\]
\item
  (Filtering consistency) For a sequence of \(\lambda_t \to 0\) and a bounded sequence of positive integers \(S_t\), the following holds: \[\underset{\bm{\theta}_{1:(t-1)}}{\sup}||\pi_t - P^{S_t}_t(\bm{\theta}_{1:(t-1)},\cdot)||_1 \leq 2\lambda_t.\]
\end{enumerate}

Let \(\epsilon_t = \rho_t^{m_t}\) and let \(Q_t = T_t^{m_t} \circ P_t^{S_t}\) be a Generative Filtering update at time \(t\). Then for any initial distribution \(\pi_0\),
\begin{equation*}||Q_t \circ \cdots \circ Q_1 \circ \pi_0 - \pi_t||_1 \leq \sum_{v=1}^t\left\{\prod_{u=v+1}^t \epsilon_u(1-\lambda_u)\right\}\epsilon_v\lambda_v \leq \sum_{v=1}^t\left\{\prod_{u=v}^t \epsilon_u\right\}\lambda_v.\end{equation*}
\end{theorem}

\begin{proof}
The proof follows closely the proof of Theorem 3.9 in \citet{yang2013sequential}.

We will construct two time inhomogeneous Markov chains \(\{X_{t,r}:r=1,\dots,m_t, t \geq 0\}\) and \(\{X_{t,r}^\prime:r=1,\dots,m_t, t \geq 0\}\). The chains proceed in the double index first in \(r\) then \(t\), i.e., \((t,r) = (0,1),\dots,(0,m_0),(1,1),\dots,(1,m_1),\dots\). The two chains are constructed as follows:

\begin{enumerate}
\def\labelenumi{\arabic{enumi}.}
\tightlist
\item
  \(X_{0,1} \sim \pi_0\), \(X^\prime_{0,1} \sim \pi_0\).
\item
  For \(t \geq 1\)

  \begin{enumerate}
  \def\labelenumii{\alph{enumii}.}
  \tightlist
  \item
    For \(r=1\). Let \(X_{t-1,m_{t-1}} = x\), \(X^\prime_{t-1,m_{t-1}} = x^\prime\). Draw \(X_{t,1} = x^\ast \sim P^{S_t}_t(x, \cdot)\). With probability \(\min\left\{1,\frac{\pi_t(x^\ast)}{P_t^{S_t}\circ\pi_{t-1}(x^\ast)}\right\}\), set \(X^\prime_{t,1}=x^\ast\); with probability \(1 - \min\left\{1,\frac{\pi_t(x^\ast)}{P_t^{S_t}\circ\pi_{t-1}(x^\ast)}\right\}\), draw \begin{equation}X^\prime_{t,1} \sim \frac{\pi_t(\cdot) - \min\left\{\pi_t(\cdot),P^{S_t}_t\circ\pi_{t-1}(\cdot)\right\}}{\tilde{\alpha}_t},\label{eqn:x-transition-between}\end{equation} where \(\tilde\alpha_t = \frac{1}{2}||\pi_t - P^{S_t}_t\circ\pi_{t-1}||_1\).
  \item
    For \(1 < r \leq m_t\). Let \(X_{t,r-1}=x\) and \(X^\prime_{t,r-1}=x^\prime\).

    \begin{enumerate}
    \def\labelenumiii{\roman{enumiii}.}
    \tightlist
    \item
      If \(x=x^\prime\), choose \(X_{t,r}=X^\prime_{t,r}\sim T_t(x,\cdot)\);
    \item
      else, first choose \(X^\prime_{t,r}=y\sim T_t(x^\prime, \cdot)\), then with probability \(\min\left\{1,\frac{T_t(x,y)}{\pi_t(y)}\right\}\), set \(X_{t,s}=y\), with probability \(1 - \min\left\{1,\frac{T_t(x,y)}{\pi_t(y)}\right\}\), draw \begin{equation}X_{t,r} \sim \frac{T_t(x,\cdot) - \min\{T_t(x,\cdot), \pi_t(\cdot)\}}{\delta_t(x)},\label{eqn:x-transition-within}\end{equation} where \(\delta_t(x) = \frac{1}{2}||T_t(x,\cdot) - \pi_t||_1\).
    \end{enumerate}
  \end{enumerate}
\end{enumerate}

First, for \(t \geq 1\) and \(1<r\leq m_t\), both chains have the same transition kernel, \(T_t\), which targets \(\pi_t\). This is apparent for \(\{X^\prime\}_{t,s}\), while for \(\{X\}_{t,s}\), we can see that its transition kernel is a mixture of \(\pi_t\) and the distribution given by (\ref{eqn:x-transition-within}) which equals \(T_t\). For \(t \geq 1\) and \(r = 1\), the distribution of \(X^\prime_{t,1}\) is \(\pi_t\) because its distribution is a mixture of \(P^{S_t}_t \circ \pi_{t-1}\) and the distribution given by (\ref{eqn:x-transition-between}), which equals \(\pi_t\). Therefore for any \((t,r)\), the marginal distribution of \(X^\prime_{t,r}\) is \(\pi_t\).

For any \((t, r)\), the marginal distribution of \(X_{t,r}\) is \(T_t^r \circ P^{S_t}_t \circ Q_{t-1} \circ \cdots \circ Q_1 \circ \pi_0\). Therefore,
\begin{equation}||Q_t \circ \cdots \circ Q_1 \circ \pi_0 - \pi_t||_1 \leq P(X_{t,m_t} \neq X^\prime_{t, m_t}).\end{equation}

Conditional on \(X_{t-1,m_{t-1}} = X^\prime_{t-1,m_{t-1}}\), the distribution of \(X_{t-1,m_{t-1}}\) is \(\pi_{t-1}\). So \(P(X_{t,1} \neq X^\prime_{t,1}|X_{t-1,m_{t-1}} = X^\prime_{t-1,m_{t-1}}) = \tilde\alpha_t\), which by Lemma \ref{lem:lemma38a}, \(\tilde\alpha_t \leq \lambda_t\).

Then,
\begin{align}
P(X_{t,m_t} \neq X^\prime_{t,m_t}) &= P(X_{t-1,m_{t-1}} \neq X^\prime_{t-1,m_{t-1}}, X_{t,m_t} \neq X^\prime_{t,m_t}) \\
&\quad + P(X_{t-1,m_{t-1}} = X^\prime_{t-1,m_{t-1}}, X_{t,m_t} \neq X^\prime_{t,m_t}) \\
&= P(X_{t,m_t} \neq X^\prime_{t,m_t} | X_{t-1,m_{t-1}} \neq X^\prime_{t-1,m_{t-1}}) \cdot P(X_{t-1,m_{t-1}} \neq X^\prime_{t-1,m_{t-1}}) \label{eqn:notequal-case-pre} \\
&\quad + P(X_{t,m_t} \neq X^\prime_{t,m_t} | X_{t-1,m_{t-1}} = X^\prime_{t-1,m_{t-1}}) \cdot (1 - P(X_{t-1,m_{t-1}} \neq X^\prime_{t-1,m_{t-1}})) \label{eqn:equal-case-pre}\\
&\leq \rho_t^{m_t} \cdot P(X_{t-1,m_{t-1}} \neq X^\prime_{t-1,m_{t-1}}) \label{eqn:notequal-case-post} \\
&\quad + \tilde{\alpha}_t\rho_t^{m_t} \cdot (1 - P(X_{t-1,m_{t-1}} \neq X^\prime_{t-1,m_{t-1}})) \label{eqn:equal-case-post} \\
&\leq \rho_t^{m_t} \cdot P(X_{t-1,m_{t-1}} \neq X^\prime_{t-1,m_{t-1}}) \\
&\quad + \lambda_t\rho_t^{m_t} \cdot (1 - P(X_{t-1,m_{t-1}} \neq X^\prime_{t-1,m_{t-1}})) \\
&= \lambda_t\rho_t^{m_t} + (1 - \lambda_t)\rho_t^{m_t}\cdot P(X_{t-1,m_{t-1}} \neq X^\prime_{t-1,m_{t-1}}) \label{eqn:sub-pprb-convergence-lemma-result}
\end{align}

Line (\ref{eqn:notequal-case-post}) follows from line (\ref{eqn:notequal-case-pre}) and line (\ref{eqn:equal-case-post}) follows from line (\ref{eqn:equal-case-pre}) because \(\rho_t\) is the probability of \(X_{t,r}\) and \(X_{t,r}^\prime\) remaining unequal given that the chains are unequal at step \(r-1\), and \(P(X_{t,1} \neq X^\prime_{t,1}|X_{t-1,m_{t-1}} = X^\prime_{t-1,m_{t-1}}) = \tilde\alpha_t\) and \(P(X_{t,1} \neq X^\prime_{t,1}|X_{t-1,m_{t-1}} \neq X^\prime_{t-1,m_{t-1}}) \leq 1\). Line (\ref{eqn:sub-pprb-convergence-lemma-result}) follows from Lemma \ref{lem:lemma38a}.

There is now a recursive relation ship between \(t\) and \(t-1\). We can repeat this for all \(t >= 1\), and using \(P(X_{0,m_0} \neq X^\prime_{0,m_0}) \leq 1\) and \(\epsilon_t = \rho_t^{m_t}\), we arrive at the result.
\end{proof}

\begin{theorem}
\protect\hypertarget{thm:bound-inequality-restated}{}\label{thm:bound-inequality-restated}\textbf{(Theorem \ref{thm:bound-inequality} restated)} Assume the following conditions hold:

\begin{enumerate}
\def\labelenumi{\arabic{enumi}.}
\tightlist
\item
  (Universal ergodicity) There exists \(\epsilon \in (0,1)\), such that for all \(t>0\) and \(x \in {\cal X}\), \[||T_t(x, \cdot) - \pi_t||_1 \leq 2\rho_t.\]
\item
  (Stationary convergence) The stationary distribution \(\pi_t\) of \(T_t\) satisfies \[\alpha_t = \frac{1}{2}||\pi_t - \pi_{t-1}||_1 \to 0,\] where \(\pi_t\) is the marginal posterior of \(\theta_{1:(t-1)}\) at time \(t\) in \(\alpha_t\).
\item
  (Filtering consistency) For a sequence of \(\lambda_t^{(F)} \to 0\) and a bounded sequence of positive integers \(S_t\), the following holds: \[\underset{\theta_{1:(t-1)}}{\sup}||\pi_t - P^{S_t}_t(\theta_{1:(t-1)},\cdot)||_1 \leq 2\lambda_t^{(F)}.\]
\item
  (Jumping consistency) For a sequence of \(\lambda_t^{(J)} \to 0\), the following holds: \[\underset{\theta_{1:(t-1)}}{\sup} ||\pi_t(\cdot | \theta_{1:(t-1)}) - J_t(\theta_{1:(t-1)}, \cdot)||_1 \leq 2\lambda_t^{(J)}.\]
\end{enumerate}

Let \(\epsilon_t = \rho_t^{m_t}\). Define \[\gamma^{(F)}_t = \sum_{v=1}^t\left\{\prod_{u=v+1}^t \epsilon_u(1-\lambda^{(F)}_u)\right\}\epsilon_v\lambda^{(F)}_v\] and \[\gamma^{(J)}_t = \sum_{v=1}^t\left\{\prod_{u=v}^t \epsilon_u\right\}(\lambda^{(J)}_v + \alpha_v)\] to be the bounds from Theorem \ref{thm:theorem39} and Theorem 3.9 of \citet{yang2013sequential}, respectively. If, for all \(u \leq t\), \(\lambda^{(F)}_u \leq \alpha_u + \lambda^{(J)}_u\), then \(\gamma_t^{(F)} \leq \gamma_t^{(J)}\).
\end{theorem}

\begin{proof}
Define \[\gamma^{(F)}_t = \sum_{v=1}^t\left\{\prod_{u=v+1}^t \epsilon_u(1-\lambda^{(F)}_u)\right\}\epsilon_v\lambda^{(F)}_s\] and \[\gamma^{(J)}_t = \sum_{v=1}^t\left\{\prod_{u=v}^t \epsilon_u\right\}(\lambda^{(J)}_v + \alpha_v).\] Assume that for all \(u \leq t\), \(\lambda^{(F)}_u \leq \alpha_u + \lambda^{(J)}_u\). Then we have,
\begin{align*}
\gamma^{(F)}_t &= \sum_{v=1}^t\left\{\prod_{u=v+1}^t \epsilon_u(1-\lambda^{(F)}_u)\right\}\epsilon_v\lambda^{(F)}_v \\
&\leq \sum_{v=1}^t\left\{\prod_{u=v+1}^t \epsilon_u\right\}\epsilon_v\lambda^{(F)}_v \\
&= \sum_{v=1}^t\left\{\prod_{u=v}^t \epsilon_u\right\}\lambda^{(F)}_v \\
&\leq \sum_{v=1}^t\left\{\prod_{u=v}^t \epsilon_u\right\}(\lambda^{(J)}_v + \alpha_v) \\
&= \gamma^{(J)}_t.
\end{align*}

Then \(\gamma_t^{(F)} \leq \gamma_t^{(J)}\).
\end{proof}

\begin{theorem}
\protect\hypertarget{thm:bound-recursive-inequality-restated}{}\label{thm:bound-recursive-inequality-restated}\textbf{(Theorem \ref{thm:bound-recursive-inequality} restated)} With the conditions and definitions of Theorem \ref{thm:bound-inequality}, assume \(\gamma^{(F)}_{t-1} = \gamma^{(J)}_{t-1}\) and define \(\gamma := \gamma^{(F)}_{t-1} = \gamma^{(J)}_{t-1}\). If \(\gamma < 1\) and \(\lambda^{(F)}_t \leq \frac{\alpha_t + \lambda^{(J)}_t}{1-\gamma}\), then \(\gamma^{(F)}_{t} \leq \gamma^{(J)}_{t}\). If \(\gamma \geq 1\) then \(\gamma^{(F)}_{t} \leq \gamma^{(J)}_{t}\) always.
\end{theorem}

\begin{proof}
We have the following recursive relationships for \(\gamma^{(F)}_{t}\) and \(\gamma^{(J)}_{t}\),
\begin{align}
\gamma^{(J)}_{t} &= \epsilon_t\gamma^{(J)}_{t-1}+\epsilon_t(\alpha_t + \lambda^{(J)}_t) \\
\gamma^{(F)}_{t} &= \epsilon_t(1-\lambda^{(F)})\gamma^{(F)}_{t-1}+\epsilon_t\lambda^{(F)}_t
\end{align}

Then for \(\gamma < 1\),
\begin{align}
\gamma^{(F)}_{t} &= \epsilon_t(1-\lambda^{(F)})\gamma^{(F)}_{t-1}+\epsilon_t\lambda^{(F)}_t \\
&= \epsilon_t(1-\lambda^{(F)})\gamma+\epsilon_t\lambda^{(F)}_t \\
&= \epsilon_t\lambda^{(F)}(1-\gamma)+\epsilon_t\gamma \\
&\leq \epsilon_t\frac{\alpha_t + \lambda^{(J)}_t}{1-\gamma}(1-\gamma)+\epsilon_t\gamma \\
&= \epsilon_t(\alpha_t + \lambda^{(J)}_t)+\epsilon_t\gamma \\
&= \gamma^{(J)}_{t}.
\end{align}

For \(\gamma \geq 1\),
\begin{align}
\gamma^{(F)}_{t} &= \epsilon_t(1-\lambda^{(F)})\gamma+\epsilon_t\lambda^{(F)}_t \\
&= \epsilon_t\lambda^{(F)}(1-\gamma)+\epsilon_t\gamma \\
&\leq \epsilon_t\gamma \\
&\leq \epsilon_t\gamma + \epsilon_t(\alpha_t + \lambda^{(J)}_t) \\
&= \gamma^{(J)}_{t}
\end{align}
\end{proof}

\begin{theorem}
\protect\hypertarget{thm:data-exponential-family-restated}{}\label{thm:data-exponential-family-restated}\textbf{(Theorem \ref{thm:data-exponential-family} restated)} Assume:

\begin{enumerate}
\def\labelenumi{\arabic{enumi}.}
\tightlist
\item
  The data \(\bm{y}_{t_1}\) and \(\bm{y}_{t_2}\), for all \(t_1 < t_2\), are conditionally independent given \(\bm{\theta}_{1:t_2}\).
\item
  Each \(\bm{y}_t\) is a sample of \(n_t\) i.i.d. observations \(\bm{y}_{t,i}\) for \(i=1,\dots,n_t\).
\item
  Each observation \(\bm{y}_{t,i}\) comes from an exponential family distribution.
\end{enumerate}

Then storage of the full data can be avoided through the use of sufficient statistics.
\end{theorem}

\begin{proof}
We have
\begin{equation}
p(\bm{y}_{t,i}|\bm{\theta}_{1:t}) = h(\bm{y}_{t,i})g(\bm{\theta}_{1:t})\exp\left\{\eta^\prime(\bm{\theta}_{1:t}) \cdot T(\bm{y}_{t,i})\right\},
\end{equation}
where \(h\) and \(g\) are scalar-valued functions, and \(\eta\) and \(T\) are (possibly) vector-valued functions of the same dimension. Then
\begin{equation}
U(\bm{y}_t) := \sum_{i=1}^{n_t} T(\bm{y}_{t,i})
\end{equation}
is a sufficient statistic for the distribution \(p(\bm{y}_t | \bm{\theta}_{1:t}) = \prod_{i=1}^{n_t} p(\bm{y}_{t,i}|\bm{\theta}_{1:t})\). Further, \(\dim U(\bm{y}_t) \approx \dim \bm{\theta}_{1:t}\), with \(\dim U(\bm{y}_t) \leq \dim \bm{\theta}_{1:t}\) unless the distribution is curved.

Then, any transition kernel can be computed while only storing the sufficient statistics, \(U_t\).
\end{proof}

\hypertarget{sec:pprb-within-gibbs-approximation-error}{%
\section{PPRB-within-Gibbs approximation error}\label{sec:pprb-within-gibbs-approximation-error}}

Section \ref{sec:bounds-on-pprb-approximation-error} deals only with the case when PPRB is used with a parameter space for \(\bm{\theta}\) which is not expanding. In this section, we extend these results to the case in which the parameter space expands with new data.

\begin{lemma}
\protect\hypertarget{lem:joint-lp-norm}{}\label{lem:joint-lp-norm}Let \(f_1(x)\) and \(f_2(x)\) be densities on the same measure space. Let \(f_x(y) := f(y|x)\) be a the probability distribution of \(y\) conditioned on \(x\), and define the joint distributions \(f_i(x,y) = f(y|x)f_i(x)\). Then \(\underset{x}{\inf}\|f_x(y)\|_p \cdot \|f_1(x) - f_2(x)\|_p \leq \|f_1(x,y) - f_2(x,y)\|_p \leq \underset{x}{\sup}\|f_x(y)\|_p \cdot \|f_1(x) - f_2(x)\|_p\).
\end{lemma}

\begin{proof}
\begin{align}
\|f_1(x,y) - f_2(x,y)\|_p &= \left(\int\int \left|f(y|x)f_1(x) - f(y|x)f_2(x)\right|^p\ \mathrm{d}y\mathrm{d}x \right)^{\frac{1}{p}} \\
&= \left(\int\int f_x(y)^p \left|f_1(x) - f_2(x)\right|^p\ \mathrm{d}y\mathrm{d}x\right)^{\frac{1}{p}} \\
&= \left(\int \|f_x(y)\|_p^p \left|f_1(x) - f_2(x)\right|^p\ \mathrm{d}x\right)^{\frac{1}{p}} \\
&\leq \left(\int \underset{x}{\sup}\|f_x(y)\|_p^p \left|f_1(x) - f_2(x)\right|^p\ \mathrm{d}x\right)^{\frac{1}{p}} \\
&= \underset{x}{\sup}\|f_x(y)\|_p \cdot \|f_1(x) - f_2(x)\|_p.
\end{align}
Similarly,
\begin{align}
\|f_1(x,y) - f_2(x,y)\|_p &= \left(\int \|f_x(y)\|_p^p \left|f_1(x) - f_2(x)\right|^p\ \mathrm{d}x\right)^{\frac{1}{p}} \\
&\geq \left(\int \underset{x}{\inf}\|f_x(y)\|_p^p \left|f_1(x) - f_2(x)\right|^p\ \mathrm{d}x\right)^{\frac{1}{p}} \\
&= \underset{x}{\inf}\|f_x(y)\|_p \cdot \|f_1(x) - f_2(x)\|_p.
\end{align}
\end{proof}

\begin{corollary}
\protect\hypertarget{cor:joint-l1-norm}{}\label{cor:joint-l1-norm}Let \(f_1(x)\) and \(f_2(x)\) be densities on the same measure space. Let \(f_x(y) := f(y|x)\) be a the probability distribution of \(y\) conditioned on \(x\), and define the joint distributions \(f_i(x,y) = f(y|x)f_i(x)\). Let \(p=1\). Then \(\|f_1(x,y) - f_2(x,y)\|_1 = \|f_1(x) - f_2(x)\|_1\).
\end{corollary}

Corollary \ref{cor:joint-l1-norm} shows that it is sufficient to consider only the accumulating \(L_1\) error of PPRB when interested in the accumulating \(L_1\) error of PPRB-within-Gibbs. PPRB-within-Gibbs at a time \(t\) targets a distribution proportional to \[p(\bm{\phi}|\bm{\theta},\bm{y}_1,\dots,\bm{y}_t)p(\bm{y}_t|\bm{\theta},\bm{y}_1,\dots,\bm{y}_{t-1})F^{(t-1)}_S(\bm{\theta}) \propto p(\bm{\phi}|\bm{\theta},\bm{y}_1,\dots,\bm{y}_t)A_{t},\]
while the true posterior at time \(t\) is proportional to
\[p(\bm{\phi}|\bm{\theta},\bm{y}_1,\dots,\bm{y}_t)p(\bm{\theta}|\bm{y}_1,\dots,\bm{y}_t) = p(\bm{\phi}|\bm{\theta},\bm{y}_1,\dots,\bm{y}_t)T_t.\]
Therefore the PPRB-within-Gibbs target distribution and the true posterior at time \(t\) have the form in Lemma \ref{lem:joint-lp-norm} and Corollary \ref{cor:joint-l1-norm} applies.

\hypertarget{sec:appendix-simulation-study}{%
\section{State space model simulation details}\label{sec:appendix-simulation-study}}

\hypertarget{formal-sampler-specification}{%
\subsection{Formal sampler specification}\label{formal-sampler-specification}}

We give a formal algorithm description for the PPRB-within-Gibbs (Algorithm \ref{alg:state-space-pprb-within-gibbs-2step}) and Generative Filtering (Algorithm \ref{alg:state-space-generative-filtering}) samplers used in Section \ref{sec:simulation-gaussian}. Steps (\ref{algstep:statespace-transition-start})-(\ref{algstep:statespace-transition-stop}) in Algorithm \ref{alg:state-space-generative-filtering} serve as applying the transition kernel (step (\ref{algstep:generative-filtering-transition-kernel}) in Algorithm \ref{alg:generative-filtering}).

\begin{algorithm}
\caption{PPRB-within-Gibbs for state space model in Section \ref{sec:simulation-gaussian}}\label{alg:state-space-pprb-within-gibbs-2step}
\algrenewcommand\algorithmicrequire{\textbf{Input:}}
\algrenewcommand\algorithmicensure{\textbf{Output:}}
\begin{algorithmic}[1]
\Require Posterior sample $\{\boldsymbol \theta_{1:(t-1),s}\}_{s=1}^S \sim p(\boldsymbol \theta_{1:(t-1)} | \boldsymbol y_{1:(t-1)})$
\Require Data $y_{t,i}$, for $i=1,\dots,n_t$
\Require Hyperparmeters, $\sigma^2, \phi^2$
\Ensure Updated posterior sample $\{\bm{\theta}_{1:t,s}\}_{s=1}^S \sim p(\bm{\theta}_{1:t} | \bm{y}_{1:t})$
    
    \State Draw $s^\ast \sim \text{Uniform}(\{1,\dots,S\})$
    \State Set $\boldsymbol \theta^{'}_{1:(t-1),0} = \boldsymbol \theta_{1:(t-1),s^\ast}$
    \State Draw $\theta_{t,0} \sim \text{N}(\theta^{'}_{t-1,0},\phi^2)$
    \For{$s=1,\dots,S$}
        \State Draw $s^\ast \sim \text{Uniform}(\{1,\dots,S\})$
        \State Set $\boldsymbol \theta_{1:(t-1)}^\ast = \boldsymbol \theta_{1:(t-1),s^\ast}$
        \State Set $\alpha = \min\left(\frac{\exp\left\{-(\theta_{t,s-1} - \theta_{t-1}^\ast)^2/(2\phi^2)\right\}}{\exp\left\{-(\theta_{t,s-1} - \theta_{t-1,s-1})^2/(2\phi^2)\right\}}, 1\right)$
        \State Draw $p \sim \text{Uniform}([0,1])$.
        \If{$p < \alpha$}
            \State Set $\boldsymbol \theta^{'}_{1:(t-1),s} = \boldsymbol \theta_{1:(t-1)}^\ast$
        \Else
            \State Set $\boldsymbol \theta^{'}_{1:(t-1),s} = \boldsymbol \theta^{'}_{1:(t-1),s-1}$
        \EndIf
        \State Set $V_t = (1/\phi^2 + n_t/\sigma^2)^{-1}$
        \State Set $C_t = \theta^{'}_{t-1,s}/\phi^2 + \sum_{i=1}^{n_t} y_{t,i} / \sigma^2$
        \State Draw $\theta_{t,s} \sim \text{N}(V_t C_t, V_t)$
    \EndFor

\end{algorithmic}
\end{algorithm}

\begin{algorithm}
\caption{Generative Filtering for state space model in Section \ref{sec:simulation-gaussian}}\label{alg:state-space-generative-filtering}
\algrenewcommand\algorithmicrequire{\textbf{Input:}}
\algrenewcommand\algorithmicensure{\textbf{Output:}}
\begin{algorithmic}[1]
\Require Posterior sample $\{\bm{\theta}_{1:(t-1),s}\}_{s=1}^S \sim p(\bm{\theta}_{1:(t-1)} | \bm{y}_{1:(t-1)})$
\Require Data $y_{t,i}$, for $i=1,\dots,n_t$
\Require Hyperparmeters, $\sigma^2, \phi^2$
\Require A $t \times t$ proposal covariance matrix, $\bm{\Sigma}_{rw}$
\Require $m_t \geq 1$
\Ensure Updated posterior sample $\{\bm{\theta}_{1:t,s}\}_{s=1}^S \sim p(\bm{\theta}_{1:t} | \bm{y}_{1:t})$

    \State Use Algorithm \ref{alg:state-space-pprb-within-gibbs-2step} to produce  $\{\bm{\theta}^{\ast}_{1:t,s}\}_{s=1}^S$ from $\{\bm{\theta}_{1:(t-1),s}\}_{s=1}^S$. 
    \For{$s=1,\dots,S$}
        \State Set $\bm{\theta}_{1:t,s} = \bm{\theta}^{\ast}_{1:t,s}$
        \For{$i=1,\dots,m_t$}
            \State Draw $\bm{\theta}_{1:t}^\ast \sim \text{N}(\bm{\theta}_{1:t,s}, \bm{\Sigma}_{rw})$ \label{algstep:statespace-transition-start}
            \State Set $\alpha = \min\left\{\frac{p(\bm{y}_{1:t}|\bm{\theta}_{1:t}^\ast,\sigma^2)p(\bm{\theta}_{1:t}^\ast|\phi^2)}{p(\bm{y}_{1:t}|\bm{\theta}_{1:t,s},\sigma^2)p(\bm{\theta}_{1:t,s}|\phi^2)}, 1\right\}$
            \State Draw $p \sim \text{Unif}([0,1])$
            \If{$p < \alpha$}
              \State Set $\bm{\theta}_{1:t,s} = \bm{\theta}_{1:t}^\ast$
            \EndIf \label{algstep:statespace-transition-stop}
        \EndFor
    \EndFor

\end{algorithmic}
\end{algorithm}

\hypertarget{parallel-computing-comparison}{%
\subsection{Parallel computing comparison}\label{parallel-computing-comparison}}

\begin{figure}
\centering
\includegraphics{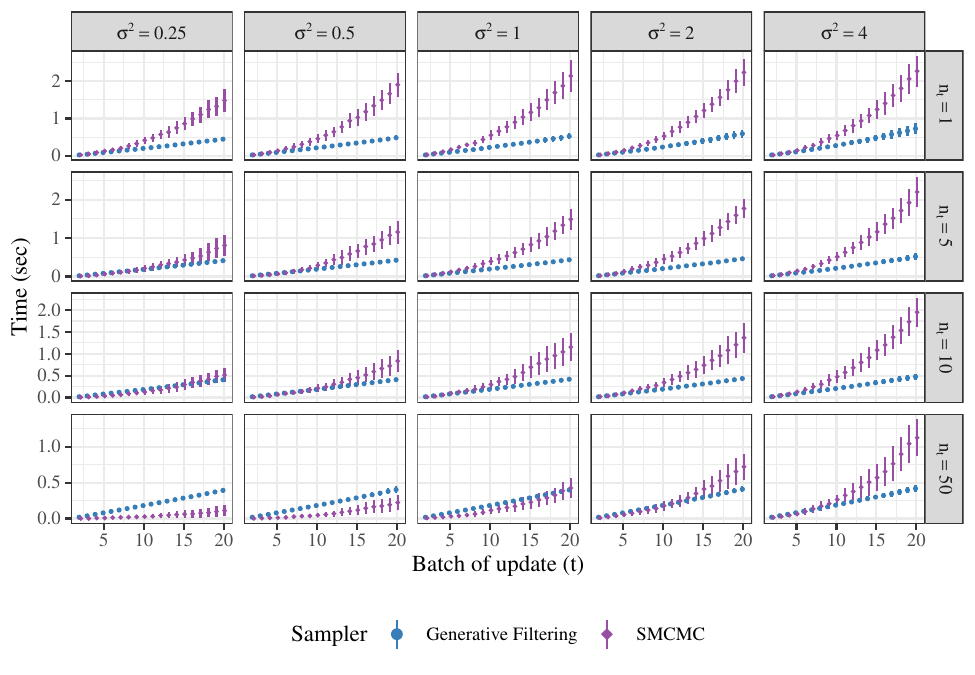}
\caption{\label{fig:simulation-plot-time}Cumulative time to converge to the posterior distribution on 8 cores. Time is shown as mean plus or minus standard deviation across all simulations. This figure illustrates the tradeoff between Generative Filtering and SMCMC. The PPRB step, while creating a better initial value, is not parallelizable. In scenarios where posterior distributions are not strongly affected by new data (e.g.~\(n_t=50, \sigma^2=0.25\)), SMCMC can converge more quickly in time because its jumping kernel is parallelizable.}
\end{figure}

The sequential nature of the PPRB-within-Gibbs filtering step of Generative Filtering presents a trade-off versus SMCMC. As seen in Figure \ref{fig:simulation-plot-steps} in Section \ref{sec:simulation-gaussian}, the filtering step initializes the Generative Filtering ensemble more effectivly than the jumping kernel of SMCMC, resulting in fewer required transition kernel steps to converge to the target distribution. However, the jumping kernel of SMCMC is parallelizable while PPRB-within-Gibbs is not. In Figure \ref{fig:simulation-plot-time} we compare the cumulative runtime required for each method to converge using 8 cores, typical of a personal workstation or laptop. When the jumping kernel and transition kernel are parallelized over 8 cores, in most scenarios, Generative Filtering takes less time to converge than SMCMC. As the number of available cores increases, this trade-off will begin to favor SMCMC in more scenarios. We see an example of this in Section \ref{sec:application}.

\hypertarget{correlation-based-online-m_t-selection}{%
\subsection{\texorpdfstring{Correlation-based online \(m_t\) selection}{Correlation-based online m\_t selection}}\label{correlation-based-online-m_t-selection}}

\citet{yang2013sequential} provide criteria for selecting the number of iterations, \(m_t\), for SMCMC in an online manner using the autocorrelation of the parallel chains. For a constant \(0 \leq \epsilon \leq 1\), the criteria selects \(m_t\) such that

\[m_t = \min\{k:\hat{f}_t(k) \leq 1 - \epsilon\}.\]
The function \(\hat{f}_t(k)\) computes the maximal correlation between the chains' values at iteration \(k\) and their initial values. For the state space model, it takes the form,

\[\hat{f}_t(k) := \underset{j=1,\dots,t}{\max}\frac{\sum_{s=1}^S\left(\theta_j^{(s,k+1)} - \overline{\theta}_j^{(k+1)}\right)\left(\theta_j^{(s,1)} - \overline{\theta}_j^{(1)}\right)}{\left(\sum_{s=1}^S\left(\theta_j^{(s,k+1)} - \overline{\theta}_j^{(k+1)}\right)^2\right)^{1/2}\left(\sum_{s=1}^S\left(\theta_j^{(s,1)} - \overline{\theta}_j^{(1)}\right)^2\right)^{1/2}},\]
where \(\theta_j^{(s,k)}\) is the value of the \(j^\text{th}\) component of \(\bm{\theta}\), in the \(k^\text{th}\) iteration of the \(s^\text{th}\) parallel chain, and \(\overline{\theta}_j^{(k)} = S^{-1}\sum_{s=1}^S \theta_j^{(s,k)}\) is the mean of the ensemble at iteration \(k\) across all chains. The initial value of each chain, \(\theta_j^{(s,1)}\), is the result of the jumping kernel (in SMCMC) or PPRB-within-Gibbs (in Generative Filtering). This criteria therefore chooses \(m_t\) such that the parameter values in each chain are sufficiently uncorrelated with their starting values.

Figure \ref{fig:simulation-plot-steps-yd} shows the cumulative number of steps required under this criteria using the recommended value of \(\epsilon = 0.5\) for small datasets \citep{yang2013sequential}. There are some key differences when compared to Figure \ref{fig:simulation-plot-steps} in Section \ref{sec:simulation-gaussian}. First, the correlation-based criteria always selects the same value of \(m_t\) for both SMCMC and Generative Filtering because it is based only on the correlation of the values at the current iteration with the initial values after the Jumping kernel or PPRB-within-Gibbs step and the same transition kernel is used for both SMCMC and Generative Filtering. Therefore it cannot take into account the accuracy of the initial ensemble at the beginning of the parallel chains. This criteria will either stop too soon for SMCMC or too late for Generative Filtering, depending on the choice of \(\epsilon\). Second, the cumulative values \(\sum_{t=1}^T m_t\) are consistent across all combinations of \(\sigma^2\) and \(n_t\), while they are not when using the oracle stopping criteria. This disparity suggests again that the correlation-based stopping criteria could select \(m_t\) that are too small or too large depending on the data, leading to either excess error or excess computational cost, respectively.

\begin{figure}
\centering
\includegraphics{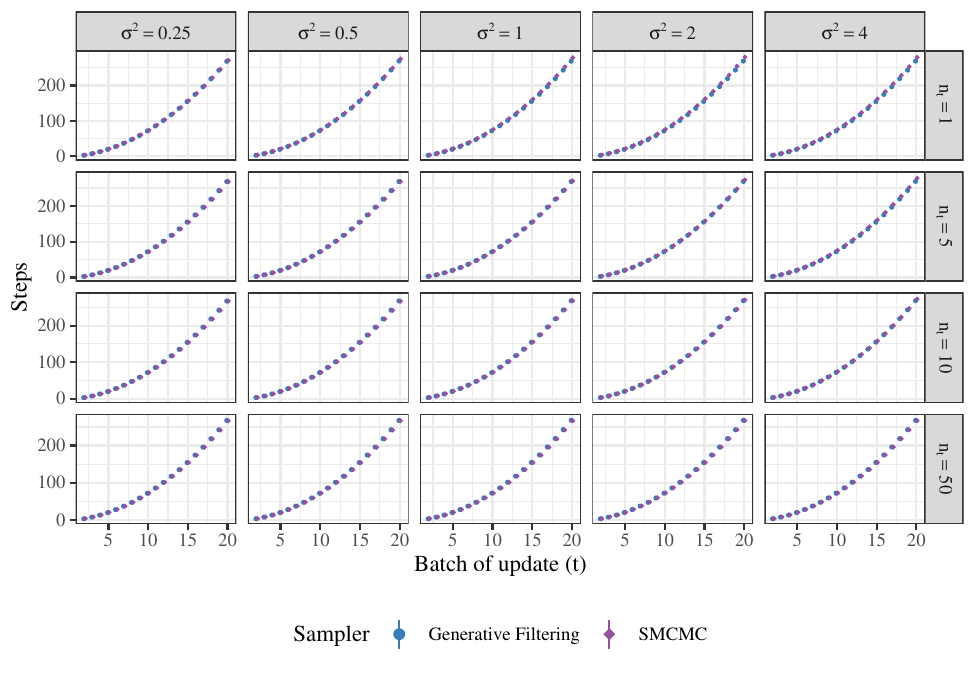}
\caption{\label{fig:simulation-plot-steps-yd}Cumulative transition kernel steps using the \citet{yang2013sequential} stopping criteria. Steps are shown as mean plus or minus standard deviation across all simulations.}
\end{figure}

Figure \ref{fig:simulation-plot-error-yd} shows the resulting MCMC error, measured by the KS statistic, when using the correlation-based stopping criteria with \(\epsilon = 0.5\). The chosen value of \(m_t\) appears sufficient for Generative Filtering, where the KS statistic is consistently below the threshold used in the oracle stopping criteria. However, in datasets with low signal from the data (\(\sigma^2\) large, \(n_t\) small, or both) the chosen value of \(m_t\) can leave significant MCMC error in the ensemble for SMCMC.

While the correlation-based stopping criteria can be inaccurate, it can serve as an approximate starting point for determining the appropriate \(m_t\) in a given problem. A stopping criteria that is flexible enough to apply to Generative Filtering and SMCMC while avoiding residual MCMC error or excessive computational cost as much as possible remains an area of future research.

\begin{figure}
\centering
\includegraphics{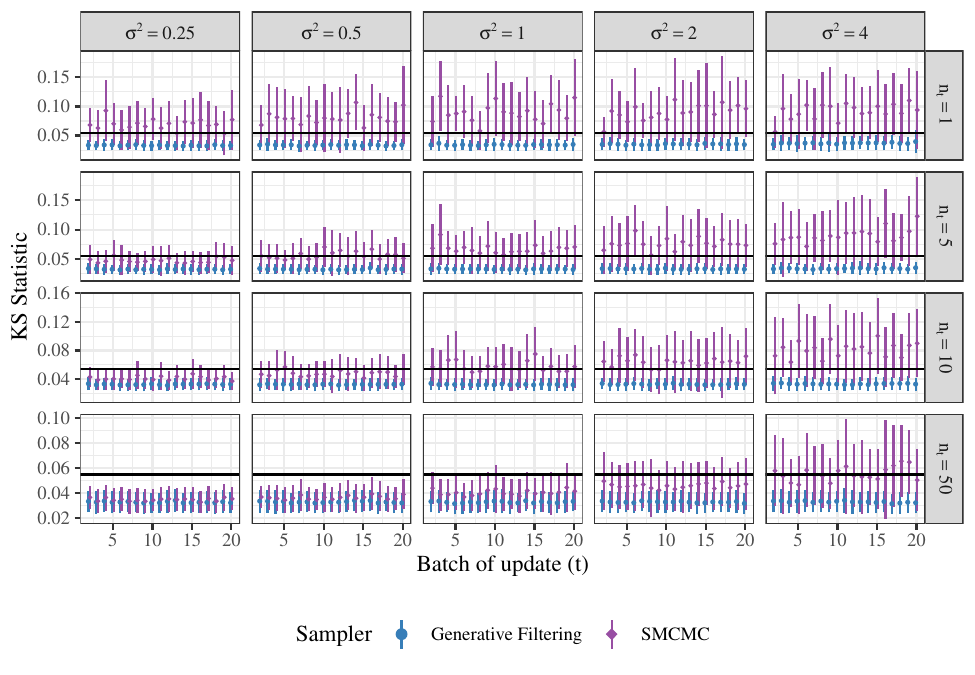}
\caption{\label{fig:simulation-plot-error-yd}Remianing MCMC error at the stopping point determined by the \citet{yang2013sequential} stopping criteria. Error shown as mean plus or minus standard deviation of KS statistic of sample. The KS critical value used in the oracle stopping criteria is shown as a horizontal line in each facet.}
\end{figure}

\hypertarget{sec:streaming-rl-model-details}{%
\section{Streaming Record Linkage Model}\label{sec:streaming-rl-model-details}}

The streaming record linkage model of \citet{taylor2023fast} is defined as follows. This section is a condensed version of the model definition in Section 2 of that paper. Consider \(k\) files \(X_1, \dots, X_k\) that are collected temporally, so that file \(X_m\) is available at time \(T_m\), with \(T_1 < T_2 < \dots < T_k\). Each file \(X_m\) contains \(n_m \geq 1\) records \(X_m = \{\boldsymbol x_{mi}\}_{i = 1}^{n_m}\), with each \(n_m\) potentially distinct. Each record is comprised of \(p_m\) fields and it is assumed that there is a common set of \(F\) fields numbered \(f=1,\dots,F\) across the \(k\) files which can be numeric, text or categorical. Records representing an individual (or entity) can be noisily duplicated across files, but it is assumed there are no duplicates within a file. We denote comparison between two records, \(\boldsymbol x_{m_1 i}\) in file \(X_{m_1}\) and \(\boldsymbol x_{m_2 j}\) in file \(X_{m_2}\), as a function, \(\gamma(\boldsymbol x_{m_1 i}, \boldsymbol x_{m_2 j})\), which compares the values in each field, \(f\), dependent on field type. Each comparison results in discrete levels \(0,\dots, L_f\) with 0 representing exact equality and subsequent levels representing increased difference. We define \(P = \sum_{f=1}^F (L_f+1)\), as the total number of levels of disagreement of all fields. The comparison \(\gamma(\boldsymbol x_{m_1 i}, \boldsymbol x_{m_2 j})\) takes the form of a \(P\)-vector of binary indicators containing \(F\) ones and \(P-F\) zeros which indicates the level of disagreement between \(\boldsymbol x_{m_1 i}\) and \(\boldsymbol x_{m_2 j}\) in each field. The comparison vectors are collected into matrices \(\Gamma^{(1)}, \dots, \Gamma^{(k-1)}\) where \(\Gamma^{(m-1)}\) contains all comparisons between the records in file \(X_m\) and previous files. Define \(\Gamma^{(1:m)}\) as \(\{\Gamma^{(1)}, \dots, \Gamma^{(m)} \}\) for \(m \in 1, \dots, k - 1\).

Records can be represented as a \(k\)-partite graph, with nodes representing records in each file and a link between two records indicating that they are coreferent. This graph can be segmented according to the order of files. First, a bipartite graph between \(X_1\) and \(X_2\); then a tripartite graph between \(X_1, X_2\), and \(X_3\), where records in \(X_3\) link to records in \(X_1\) and \(X_2\); until finally a \(k\)-partite graph between \(X_1,\dots,X_k\) where records in \(X_k\) link to records in \(X_1,\dots,X_{k-1}\). These graphs can be represented with \(k-1\) matching vectors, with one vector per file \(X_2, \dots, X_k\). Each vector, denoted \(\boldsymbol Z^{(m-1)}\), has length \(n_m\) with the value in index \(j\), denoted \(Z^{(m-1)}_j\), corresponding to the record \(\boldsymbol x_{mj}\) as follows,
\[
Z^{(m-1)}_j = \begin{cases}
\sum_{\ell=1}^{t-1} n_\ell + i & \parbox[t]{.60\textwidth}{for $t < m$, if $\boldsymbol x_{ti} \in X_t$ and $\boldsymbol x_{mj}$ are coreferent,} \\ 
\sum_{\ell=1}^{m-1} n_\ell +j & \text{otherwise.}
\end{cases}
\]
Let \(\boldsymbol Z^{(m-1)} = \left(Z^{(m-1)}_j\right)_{j=1}^{n_{m}}\) and \(\boldsymbol Z^{(1:m)} = \left\{\boldsymbol Z^{(1)},\dots,\boldsymbol Z^{(m)}\right\}\) for \(m \in 1, \dots, k - 1\). These vectors identify which records are coreferent and are therefore the main parameters of interest in the record linkage problem.

Preserving the assumption of duplicate-free files with a large number of files is a challenge because the combination of several links throughout the parameters \(\boldsymbol Z^{(1:(k-1))}\) may imply that two records in the same file are coreferent. To enforce transitivity of the coreference relationship, comparisons with files \(X_m, m \geq 3\) will be constrained.

\begin{definition}
\protect\hypertarget{def:link-validity}{}\label{def:link-validity}\textbf{Link Validity Constraint.} Let \(\mathcal{C}_k\) be the set of all matching vectors \(\boldsymbol Z^{(1:(k-1))}\) such that every record \(\boldsymbol x_{m_1 i}\) receives at most one link from a record \(\boldsymbol x_{m_2 j}\) where \(m_2 > m_1\). That is, there is at most one value in any \(\boldsymbol Z^{(m_2-1)}\) with \(m_2 > m_1\) that equals \(\sum_{\ell=1}^{m_1 - 1} n_\ell + i\). Matching vectors \(\boldsymbol Z^{(1:(k-1))}\) are valid if and only if \(\boldsymbol Z^{(1:(k-1))} \in \mathcal{C}_k\).
\end{definition}

We also define parameters \(\boldsymbol m\) and \(\boldsymbol u\), which fully specify the distributions \(\cal M\) and \(\cal U\) respectively. Both \(\boldsymbol m\) and \(\boldsymbol u\) are \(P\)-vectors which can be separated into the sub-vectors \(\boldsymbol m = \begin{bmatrix}\boldsymbol m_1 & \dots & \boldsymbol m_F\end{bmatrix}\) and \(\boldsymbol u = \begin{bmatrix}\boldsymbol u_1 & \dots & \boldsymbol u_F\end{bmatrix}\), where \(\boldsymbol m_f\) and \(\boldsymbol u_f\) have length \(L_f+1\). Then \({\cal M}(\boldsymbol m) = \prod_{f=1}^F \text{Multinomial}(1; \boldsymbol m_f)\) and \({\cal U}(\boldsymbol u) = \prod_{f=1}^F \text{Multinomial}(1; \boldsymbol u_f)\) are the distributions for matches and non-matches, respectively. To define the likelihood, we first define the match set,
\(M := M(\boldsymbol Z^{(1:(k-1))}) = \{(\boldsymbol x_{m_1 i}, \boldsymbol x_{m_2 j}): \boldsymbol x_{m_1 i}\text{ and }\boldsymbol x_{m_2 j}\text{ are linked}\},\) to contain all pairs of records that are linked either directly or transitively through a combination of multiple vectors \(\boldsymbol Z^{(1:(k-1))}\).

The full data model in the \(k\)-file case is then
\begin{equation}
P(\Gamma^{(1:(k-1))} | \boldsymbol m, \boldsymbol u, \boldsymbol Z^{(1:(k-1))}) = \prod_{m_1 < m_2}^k \prod_{i=1}^{n_{m_1}} \prod_{j=1}^{n_{m_2}} \prod_{f=1}^F \prod_{\ell=0}^{L_f} \left[m_{f\ell}^{\mathbb{I}((\boldsymbol x_{m_1 i}, \boldsymbol x_{m_2 j}) \in M)} u_{f\ell}^{\mathbb{I}((\boldsymbol x_{m_1 i}, \boldsymbol x_{m_2 j}) \notin M)} \right]^{\gamma^{f\ell}(\boldsymbol x_{m_1 i}, \boldsymbol x_{m_2 j})}. \label{eqn:streaming-data-model}
\end{equation}
The support of the data distribution is dependent on the vectors \(\boldsymbol Z^{(1:(k-1))}\), specifically, the matching vectors must satisfy the link validity constraint given in Definition \ref{def:link-validity}. We explicitly write this constraint as an indicator function in the likelihood:
\begin{equation}
L(\boldsymbol m, \boldsymbol u, \boldsymbol Z^{(1:(k-1))}) = \mathbb{I}(\boldsymbol Z^{(1:(k-1))} \in \mathcal{C}_k)\cdot P(\Gamma^{(1:(k-1))} | \boldsymbol m, \boldsymbol u, \boldsymbol Z^{(1:(k-1))}). \label{eqn:streaming-likelihood}
\end{equation}

The parameters \(\boldsymbol m\) and \(\boldsymbol u\) are probabilities of a multinomial distribution, so we specify conjugate Dirichlet priors. Specifically, we let \(\boldsymbol m_f \sim \text{Dirichlet}(\boldsymbol a_f)\) and \(\boldsymbol u_f \sim \text{Dirichlet}(\boldsymbol b_f)\), for \(f=1,\dots,F\), where \(\boldsymbol a_f\) and \(\boldsymbol b_f\) are vectors with the same dimension, \(L_f+1\), as \(\boldsymbol m_f\) and \(\boldsymbol u_f\). To define the prior for the matching vectors, let \(w_j^{(k)} := \mathbb{I}\left(Z^{(k-1)}_j \leq \sum_{m=1}^{k-1} n_m\right)\), that is let \(w_j^{(k)}\) be an indicator that record \(j\) in file \(k\) is linked, and \(\boldsymbol w^{(k)} = \left\{w_j^{(k)}: j = 1, \dots, n_k\right\}\). Then to specify the prior for \(\boldsymbol Z^{(k-1)}\), let
\begin{align}
\left. w_j^{(k)} \middle| \pi \right. &\stackrel{\text{iid}}{\sim} \text{Bernoulli}(\pi) \notag \\
\left. \boldsymbol Z^{(k-1)} \middle| \boldsymbol w^{(k)}\right. &\sim \text{Uniform}\left(\left\{\text{all valid $k$-partite matchings}\right\}\right). \label{eqn:z-prior}
\end{align}

Allowing \(\pi \sim \text{Beta}(\alpha_\pi, \beta_\pi)\) results in the marginal streaming prior
\begin{equation*}
P(\boldsymbol Z^{(k-1)}|\alpha_\pi,\beta_\pi) = \frac{(N - n_{k\cdot}(\boldsymbol Z^{(k-1)}))!}{N!}\cdot\frac{\mbox{B}(n_{k\cdot}(\boldsymbol Z^{(k-1)}) + \alpha_\pi, n_k - n_{k\cdot}(\boldsymbol Z^{(k-1)}) + \beta_\pi)}{\mbox{B}(\alpha_\pi, \beta_\pi)}, 
\end{equation*}
where \(N = \sum_{m=1}^{k-1} n_m\) and \(n_{k\cdot}(\boldsymbol Z^{(k-1)}) = \sum_{j=1}^{n_k} I(Z^{(k-1)}_j \leq N)\).

\hypertarget{sec:application-sampling-details}{%
\section{Pups Sampling Details}\label{sec:application-sampling-details}}

We model the sea lion pup counts using the following hierarchical model,
\begin{align*}
y_{s,t} &\sim \mbox{Pois}(\lambda_{s,t}) \\
\log(\lambda_{s,1}) &\sim \mbox{N}(\mu_1, \sigma_1^2) \\
\log(\lambda_{s,t}) &\sim \mbox{N}(\phi_s + \log(\lambda_{s,t-1}), \sigma^2_s) \\
\phi_s &\sim \mbox{N}(0, \sigma^2_\phi) \\
\sigma^2_s &\sim \mbox{Inverse-gamma}(\alpha, \beta),
\end{align*}
where \(s=1,2,3,4\) for each of our four studied sites, and \(t=1978, \dots, 2016\). The parameters \(\lambda_{s,t}\) represent the latent population intensities, and the parameters \(\phi_s\) and \(\sigma^2_s\) define the relationship between population intensity parameters over time. The parameters of interest are population intensities \(\lambda_{s,t}\) and population intensity trends \(\phi_s\). Negative values of \(\phi_s\) indicate the population intensities are decreasing at site \(s\) while positive values of \(\phi_s\) indicate the population intensities are increasing at site \(s\). We set the hyperparameters \(\mu_1 = 8.7\), \(\sigma^2_1 = 1.69\), \(\sigma^2_\phi = 1\), \(\alpha = 1\), and \(\beta = 20\).

We estimate the parametrers using a Gibbs sampler, SMCMC, PPRB-within-Gibbs (Algorithm \ref{alg:pprb-within-gibbs-2step}), and Generative Filtering (Algorithm \ref{alg:generative-filtering}). The Gibbs sampler uses conjugate full conditional updates for \(\phi_s\) and \(\sigma^2_s\), and Metropolis-within-Gibbs proposals for \(\log(\lambda_{s,t})\), as described in \citet{hooten2021making}. We reproduce these full conditional distributions here for convenience. When data, \(\bm{y} = \{y_{s,t}\}_{s=1,\dots,4; t=1,\dots,T}\), have arrived for a time \(T\). Then each parameter has full conditional distributions,
\begin{align*}
\phi_s|\cdot &\sim \mathrm{N}(a^{-1}b, a^{-1}), \\
\sigma^2_s|\cdot &\sim \mathrm{IG}(\Tilde{\alpha}, \Tilde{\beta}), \\
p(\log(\lambda_{s,t})|\cdot) &\propto \begin{cases}
p(y_{s,1}|\lambda_{s,1})p(\log(\lambda_{s,2})|\phi_s,\sigma_s^2,\log(\lambda_{s,1}))p(\log(\lambda_{s,1})) & \text{for $t=1$}, \\
p(y_{s,t}|\lambda_{s,t})p(\log(\lambda_{s,t+1})|\phi_s,\sigma_s^2,\log(\lambda_{s,t}))p(\log(\lambda_{s,t})|\phi_s,\sigma_s^2,\log(\lambda_{s,t-1})) & \text{for $1<t<T$}, \\
p(y_{s,T}|\lambda_{s,T})p(\log(\lambda_{s,T})|\phi_s,\sigma_s^2,\log(\lambda_{s,
T-1})) & \text{for $t=T$},
\end{cases}
\end{align*}
where
\begin{align*}
a &= \frac{T-1}{\sigma_s^2} + \frac{1}{\sigma_\phi^2}, \\
b &= \frac{1}{\sigma_s^2}\left(\sum_{t=2}^T (\log(\lambda_{s,t}) - \log(\lambda_{s,t-1})) \right) = \frac{\log(\lambda_{s,T}) - \log(\lambda_{s,1})}{\sigma_s^2}, \\
\Tilde{\alpha} &= \frac{T-1}{2} + \alpha, \\
\Tilde{\beta} &= \left(\frac{\sum_{t=2}^T (\log(\lambda_{s,t}) - \phi_s - \log(\lambda_{s,t-1}))^2}{2}+\frac{1}{\beta}\right)^{-1}.
\end{align*}

Algorithm \ref{alg:pups-gibbs-sampler} gives the full algorithm for producing posterior samples of the parameters in this model. We choose the proposal variances \(\sigma^2_{\mathrm{tune},s,t}\) such that for each \(s,t\) the proposals will have an acceptance rate of approximately 0.44. We refer to steps \ref{algstep:pups-gibbs-kernel-begin} through \ref{algstep:pups-gibbs-kernel-end} in Algorithm \ref{alg:pups-gibbs-sampler} as the Gibbs-style transition kernel, which updates all parameters in a manner such that the posterior is the stationary distribution. This kernel is used as the transition kernel for Generative Filtering (step \ref{algstep:generative-filtering-transition-kernel} of Algorithm \ref{alg:generative-filtering}) and SMCMC. The random walk Metropolis proposal for \(\log(\lambda_{s,T})\) (steps \ref{algstep:pups-lambda-fc-begin} through \ref{algstep:pups-lambda-fc-end} of Algorithm \ref{alg:pups-gibbs-sampler}, when \(t=T\)) is used to update \(\log(\lambda_{s,T})\) in the PPRB-within-Gibbs sampler (step \ref{algstep:pprb-within-gibbs-phi-update} of Algorithm \ref{alg:pprb-within-gibbs-2step}) and as the jumping kernel in SMCMC.

\spacingset{1}
\begin{algorithm}
\caption{Gibbs Sampler for sea lion pup count model}\label{alg:pups-gibbs-sampler}
\algrenewcommand\algorithmicrequire{\textbf{Input:}}
\algrenewcommand\algorithmicensure{\textbf{Output:}}
\algblock{If}{EndIf}
\algcblock[If]{If}{ElsIf}{EndIf}
\algcblock{If}{Else}{EndIf}
\begin{algorithmic}[1]
\Require Data $y_{s,t}$ for $s=1,\dots,4$ and $t=1,\dots,T$
\Require Proposal variances $\sigma^2_{\mathrm{tune},s,t}$ for $s=1,\dots,4$ and $t=1,\dots,T$
\Require Hyperparameter values $\mu_1, \sigma^2_1, \sigma^2_\phi, \alpha, \beta$
\Require Desired number of samples, $N$
\Require Initial parameter values $(\phi_1,\dots,\phi_4,\sigma^2_1,\dots,\sigma^2_4,\lambda_{1,1},\dots,\lambda_{4,T})^{(0)}$
\Ensure Posterior sample $\{(\phi_1,\dots,\phi_4,\sigma^2_1,\dots,\sigma^2_4,\lambda_{1,1},\dots,\lambda_{4,T})^{(i)}\}_{i=1}^N$

  \For{$i=1,\dots,N$}
    \For{$s=1,\dots,4$} \label{algstep:pups-gibbs-kernel-begin}
      \State Set $a = \frac{T-1}{{\sigma_s^2}^{(i-1)}} + \frac{1}{\sigma_\phi^2}$, $b = \frac{\log(\lambda_{s,T}^{(i-1)}) - \log(\lambda_{s,1}^{(i-1)})}{{\sigma_s^2}^{(i-1)}}$.
      \State Draw $\phi_s^{(i)} \sim \mathrm{N}(a^{-1}b, a^{-1})$.
      \State Set $\Tilde{\alpha} = \frac{T-1}{2} + \alpha$, $\Tilde{\beta} = \left(\frac{\sum_{t=2}^T (\log(\lambda_{s,t}^{(i-1)}) - \phi_s^{(i)} - \log(\lambda_{s,t-1}^{(i-1)}))^2}{2}+\frac{1}{\beta}\right)^{-1}$.
      \State Draw ${\sigma^2_s}^{(i)} \sim \mbox{Inverse-gamma}(\Tilde{\alpha}, \Tilde{\beta})$.
      \For{$t=1,\dots,T$}
        \State Draw $\log(\lambda_{s,t}^{'}) \sim \mathrm{N}(\log(\lambda_{s,t}^{(i-1)}), \sigma^2_{\mathrm{tune},s,t})$. \label{algstep:pups-lambda-fc-begin}
        \If{$t=1$}
          \State Set $\alpha = \min\left\{\frac{p(y_{s,1}|\lambda_{s,1}^{'})p(\log(\lambda_{s,2}^{(i-1)})|\phi_s^{(i)},{\sigma_s^2}^{(i)},\log(\lambda_{s,1}^{'}))p(\log(\lambda_{s,1}^{'}))}{p(y_{s,1}|\lambda_{s,1}^{(i-1)})p(\log(\lambda_{s,2}^{(i-1)})|\phi_s^{(i)},{\sigma_s^2}^{(i)},\log(\lambda_{s,1}^{(i-1)}))p(\log(\lambda_{s,1}^{(i-1)}))}, 1\right\}$
        \ElsIf{$t=T$}
          \State Set $\alpha = \min\left\{\frac{p(y_{s,T}|\lambda_{s,T}^{'})p(\log(\lambda_{s,T}^{'})|\phi_s^{(i)},{\sigma_s^2}^{(i)},\log(\lambda_{s,T-1}^{(i)}))}{p(y_{s,T}|\lambda_{s,T}^{(i-1)})p(\log(\lambda_{s,T}^{(i-1)})|\phi_s^{(i)},{\sigma_s^2}^{(i)},\log(\lambda_{s,T-1}^{(i)}))}, 1\right\}$
        \Else
          \State Set $\alpha = \min\left\{\frac{p(y_{s,t}|\lambda_{s,t}^{'})p(\log(\lambda_{s,t+1}^{(i-1)})|\phi_s^{(i)},{\sigma_s^2}^{(i)},\log(\lambda_{s,t}^{'}))p(\log(\lambda_{s,t}^{'})|\phi_s^{(i)},{\sigma_s^2}^{(i)},\log(\lambda_{s,t-1}^{(i)}))}{p(y_{s,t}|\lambda_{s,t}^{(i-1)})p(\log(\lambda_{s,t+1}^{(i-1)})|\phi_s^{(i)},{\sigma_s^2}^{(i)},\log(\lambda_{s,t}^{(i-1)}))p(\log(\lambda_{s,t}^{(i-1)})|\phi_s^{(i)},{\sigma_s^2}^{(i)},\log(\lambda_{s,t-1}^{(i)}))}, 1\right\}$
        \EndIf
        \State Draw $p \sim \text{Uniform}([0,1])$.
        \If{$p < \alpha$}
            \State Set $\log(\lambda_{s,t}^{(i)}) = \log(\lambda_{s,t}^{'})$
        \Else
            \State Set $\log(\lambda_{s,t}^{(i)}) = \log(\lambda_{s,t}^{(i-1)})$
        \EndIf \label{algstep:pups-lambda-fc-end}
      \EndFor
    \EndFor \label{algstep:pups-gibbs-kernel-end}
  \EndFor

\end{algorithmic}
\end{algorithm}
\spacingset{1.75}

\end{document}